%% file: analysis.tex
\documentclass[11pt]{article}

\usepackage[margin=1in]{geometry}
\usepackage{amsmath,amssymb,amsthm,mathtools}
\usepackage{graphicx}
\usepackage{tikz}
\usetikzlibrary{arrows.meta,backgrounds,calc,fit,positioning,decorations.pathreplacing}
\definecolor{hc}{HTML}{297F91}
\definecolor{ct}{HTML}{7155A1}
\usepackage{xcolor}
\usepackage{microtype}
\usepackage[colorlinks=true,linkcolor=blue,citecolor=blue,urlcolor=blue]{hyperref}
\usepackage{cleveref}
\usepackage{braket}
\usepackage{physics}
\usepackage{booktabs}
\usepackage{subcaption}
\usepackage[backend=bibtex8,doi=false,isbn=false,url=false,maxbibnames=5,citestyle=numeric-comp,sorting=none,
sortcites=true,
style=numeric-comp]{biblatex}
\bibliography{joel}

\newtheorem{theorem}{Theorem}[section]
\newtheorem{lemma}[theorem]{Lemma}
\newtheorem{corollary}[theorem]{Corollary}
\newtheorem{proposition}[theorem]{Proposition}
\newtheorem{definition}[theorem]{Definition}

\newtheorem{fact}[theorem]{Fact}
\numberwithin{equation}{section}
\usepackage{xcolor}
\usepackage[linesnumbered,ruled,vlined]{algorithm2e}

\definecolor{hard}{HTML}{8B0A50}
\definecolor{easy}{HTML}{1F3F8F}
\definecolor{wire}{HTML}{52616B}
\newcommand{\hard}[2]{\textcolor{hard!#1}{#2}}   
\newcommand{\easy}[2]{\textcolor{easy!#1}{#2}}

\newcommand{\mc}{\mathcal}

\begin{document}

\title{A polynomial-time classical sampler for noisy quantum circuits from statistical mechanics}



\author{%
  Jon Nelson\textsuperscript{1}%
  \thanks{Email: \texttt{nelson1@umd.edu}}
  \footnotemark[3]
  \quad Joel Rajakumar\textsuperscript{1,2} %
    \thanks{Email: \texttt{jrajakum@ucla.edu}}
  %
  \thanks{These authors contributed equally to this work.}
  \\
  Chao Yin\textsuperscript{4}%
    \thanks{Email: \texttt{chaoyin@stanford.edu}}
  \quad Yifan F. Zhang\textsuperscript{5}%
    \thanks{Email: \texttt{yz4281@princeton.edu}}\\
  \quad Michael J. Gullans\textsuperscript{1,5} \thanks{Email: \texttt{mgullans@umd.edu}}\\\\
  {\small \textsuperscript{1}University of Maryland,
    College Park, Maryland, USA}\\
  {\small \textsuperscript{2}Electrical \& Computer Engineering Department, UCLA, Los Angeles,
    California, USA}\\
  {\small \textsuperscript{3}Department of Physics,
    Stanford University, Stanford, California, USA}\\
  {\small \textsuperscript{4}Department of Electrical and
    Computer Engineering, Princeton University,  Princeton, New Jersey, USA}\\
    {\small \textsuperscript{5}QuEra Computing Inc.,
    Boston, Massachusetts, USA}\\
}
\date{}
\maketitle

\begin{abstract}
Developing classical simulation algorithms for noisy quantum circuits is essential to delineating the limits of quantum advantage. Existing classical sampling approaches for general circuits require circuit depths to grow logarithmically with system size, so that noise drives the global output state close to a trivial state. Here we show that local accumulation of noise at a depth independent of system size is already sufficient. Specifically, we prove that any geometrically local circuit composed of unital operations and interspersed with single-qubit depolarizing noise of strength $p$ after $\Theta(p^{-1} \log p^{-1})$ depth can be approximately sampled from a polynomial-time classical computer. This generalizes existing results that impose anticoncentration assumptions or non-universal gate sets in order to obtain a classical sampler at such shallow depths. Our proof maps the output state of the circuit to a polymer model in statistical mechanics, and combines a convergent cluster expansion with hypercontractivity of the depolarizing channel.
\end{abstract}


\section{Introduction} \label{sec:intro}

It is well-established that ideal quantum computers can exponentially outperform classical computers on tasks like factoring and quantum simulation \cite{shor_factoring_1997,lloyd_quantum_simulators_1996}. However, in reality, experimental quantum computers constantly experience noise, which disrupts the computation. Fortunately, the celebrated `threshold theorem' \cite{aharonov_ben-or,knill_noisy_devices_2005}  suggests that if the noise strength remains below a threshold, a noisy quantum computer can reproduce any computation performed by an ideal quantum computer. An essential feature of these fault tolerance protocols is \textit{non-unital} operations enabled by mid-circuit measurements or the ability to reset/reload qubits. Experimental progress toward fault tolerant quantum computers has been promising \cite{bluvstein_logical,google_threshold,quantinuum_arbitrary_geometries,ibm_evidence_for_utility}, but most devices can only natively support \textit{unital} operations on a plane with sparse or costly access to
non-unital operations
\cite{geher_reset_2025,decross_qubit_reuse_2023,chiu_continuous_operation_3000-qubit_system}. Nevertheless, there is evidence that such devices can perform finite-size computations which are out of reach for current classical algorithms \cite{google_phase_transitions,quantinuum_arbitrary_geometries,hangleiter_fault-tolerant}. An important forward-looking question is:\\

\textit{When and how will quantum advantages break down if experimental quantum devices are scaled in size and runtime without also scaling non-unital/non-local resources?} \\

The only known mechanism by which noise (e.g. local depolarizing channels) removes quantum advantage from arbitrary unital quantum circuits is the gradual converence of the output distribution to uniform \cite{aharonov_limitations,muller-hermes_stilck-franca_wolf_relative_entropy_convergence,nelson_rajakumar_limitations}. This establishes classical simulability when circuit depth grows \textit{superlogarithmically} with system size, but leaves open a shallow depth regime with the possibility of quantum advantage, which has been studied extensively in prior work \cite{fujii_noise_threshold,bouland_fefferman_landau_liu_noise,bremner_achieving,rajakumar_watson,anshu_liu_nguyen_pattison_loglogn_iqp,aharonov_limitations}. Further negative results in this shallow regime have only managed to rule out certain \textit{specific} candidates for quantum advantage such as estimation of observables \cite{schuster_polynomial,garcia_cirac_trivedi_pauli,fontana,angrisani_mele_simulating,mele_noise-induced,martinez_simulation,stilck-franca_garcia-patron,de-palma_marvian_rouze_stilck-franca_optimal_transport}, sampling from anti-concentrated output distributions \cite{aharonov_polynomial_2023,bremner_achieving,gao_duan,takahashi_ct-ecs,schuster_polynomial,lee_ghosh_approx_markov}, or executing circuits with non-universal gate sets \cite{nelson_rajakumar_naturally,rajakumar_watson_liu,oh_classical_simulability_linear_optical,park_oh_mps_noisy_iqp,shravan_iqp_damping}. Our work instead provides a definite no-go result for quantum advantage by introducing techniques to classically sample from so-called `worst-case' noisy quantum circuits at low depths, before convergence to uniformity. We show that for \textit{any} quantum circuit with depth greater than a \textit{constant} critical threshold, its output distribution admits an efficient classical sampling algorithm. Thus, we close an existing gap, showing rigorously that building larger and larger quantum devices will not enable any additional quantum computational depth, if non-unital/non-local resources are not also scaled in turn. We also remark that our bounds may apply even to future fault tolerant devices that implement logical unitary gates with logical depolarizing-like noise and some notion of geometric locality.

\begin{theorem}[Informal]\label{thm:informal}
Given the description of any quantum circuit with $d$ layers of $2$-qubit unital operations applied in parallel to qubits that are adjacent in an $O(1)$-dimensional lattice and initialized to $\ket{0}$, if single-qubit depolarizing channels
$\mathcal{N}^{p}(\rho) = (1-p)\rho + p\,\tfrac{\mathbb{I}}{2}$ are introduced to each qubit after each
layer, a polynomial-time classical algorithm can sample from its output distribution with inverse-polynomial total variation distance when $d > d_{crit} = O(p^{-1}\log p^{-1})$.
\end{theorem}

%
%
%

Beyond its practical relevance, a further motivation for our work is that noisy quantum circuits broadly encompass much of the dynamics that arise in nature. Noise naturally models the gradual growth of entropy (heat) in subsystems that are not completely isolated from their environment, and non-unital processes in this subsystem require further coupling to a low-temperature bath (e.g. a classical computer performing error correction) which actively or passively drains heat out~\cite{ben-or_gottesman_refrigerator,dennis_topological_quantum_memory_2002,alicki_thermal_stability_2010,shtanko_sharma_nonunital_1D}. It is philosophically interesting to wonder what is the computational complexity of such open quantum systems `on their own,' without such a low-temperature bath \cite{yunchao_liu_thesis}. This gets at the well-studied question, \\

\textit{When do noisy many-body quantum dynamics in nature exhibit a quantum-to-classical transition?} \\

The transition shown in our work can be understood as a low-to-high-temperature transition that occurs at a critical point \cite{bakshi_high-temperature,zlokapa_quasipolynomial_syk_2026,yin_polynomial-time_classical_sampling}.
Similar ideas have appeared in earlier work~\cite{kempe_upper_bounds,aharonov_polynomial_1996,aharonov_quantum_to_classical,fujii_tamate,cheng_ippoliti,seddon_magic_monotones,plenio_upper_bounds,burhman_upper_bounds,razborov_upper_bounds,zhang_lee_approx_markov,trivedi_cirac_simulability_transitions}, but the effective temperature is characterized only by the noise strength, 
where the high-temperature regime requires
very high noise strengths, e.g. $p > p_{crit} = 0.357$ \cite{kempe_upper_bounds}. When examining how depth affects this critical point, there are two competing effects: heat locally accumulates with circuit depth, introducing an additional factor of $d$ to the effective temperature, but entanglement also grows with circuit depth, so more and more local heat is required to cross the critical point. 
We show that noise ends up winning this competition against entanglement,
establishing a classical phase when $\tilde{\Theta}(pd)$ crosses a system-size independent critical threshold. As mentioned earlier, prior results only
establish a high-temperature regime only after $\tilde{\Theta}(pd) > \log (n)$ \cite{aharonov_limitations,muller-hermes_stilck-franca_wolf_relative_entropy_convergence}\footnote{These results apply broadly to systems with long-range interactions, where the low-temperature regime can be protected by error correction up until $pd = O(\log n)$  \cite{aharonov_limitations}. It has been observed that such error correction fails under geometric locality, but connecting this to classical simulability was left open \cite{nelson_rajakumar_limitations}.}; by getting rid of the dependence on $n$, our work establishes for the first time that critical phenomena occur here at arbitrarily low noise strengths in the thermodynamic limit.


In fact, our classical algorithm is much stronger than an additive-error sampler: it is a fully polynomial-time approximation scheme (FPTAS) for arbitrary marginals of the output distribution. This produces a relative-error sampler, i.e. each output probability is within $(1\pm\epsilon)$ times the true output probability. Relative-error sampling is the only sampling task that is proven to be classically hard up to non-collapse of the polynomial hierarchy, without further unproven conjectures\footnote{Even assuming these conjectures, additive-error hardness in constant-depth noisy geometrically local circuits is an open question with no existing proposals (see \cref{sec:open-questions})} \cite{terhal_divincenzo, bremner_jozsa_shepherd, aaronson_arkhipov}. Notably, this task is classically hard for noisy circuits with depth below an $O(p^{-1})$ threshold via a postselection argument using cluster states \cite{fujii_tamate}\footnote{This result shows hardness in 3D depth-4 circuits with noise below a fixed threshold. One can pad $\Omega(p^{-1})$ noisy identity layers to this circuit and still keep the accumulated noise under this threshold.}, showing that our result is tight up to logarithmic factors.

\begin{figure}[ht]
    \centering
    \begin{tikzpicture}[
        x=1cm,
        y=1cm,
        axis/.style={line width=0.85pt},
        tick/.style={line width=0.85pt},
        depthlabel/.style={align=center, font=\scriptsize, inner sep=1pt},
        coverage/.style={draw=easy, fill=easy!18, line width=0.65pt,
                        rounded corners=0.5pt},
        ourwork/.style={draw=easy, fill=easy, line width=0.65pt,
                       rounded corners=0.5pt},
        hardness/.style={draw=hard, fill=hard!18, line width=0.65pt,
                        rounded corners=0.5pt}
    ]
        \path[hardness] (0,1.02) rectangle (2.06,1.30);
        \path[coverage] (8.54,1.02) rectangle (10.76,1.30);
        \path[coverage] (10.84,1.02) rectangle (13.30,1.30);

        \path[ourwork] (2.94,0.48) rectangle (13.30,0.76);

        \node[align=center,font=\scriptsize,text=hard] at (1.03,1.55)
            {\cite{fujii_tamate}};
        \node[align=center,font=\scriptsize\bfseries,text=easy] at (5.25,0.91)
            {this work};
        \node[align=center,font=\scriptsize,text=easy] at (9.65,1.55)
            {\cite{nelson_rajakumar_limitations}};
        \node[align=center,font=\scriptsize,text=easy] at (12.07,1.55)
            {\cite{muller-hermes_stilck-franca_wolf_relative_entropy_convergence,
             aharonov_limitations}};

        \path[hardness] (7.35,1.87) rectangle (7.68,2.07);
        \node[anchor=west,font=\scriptsize] at (7.80,1.97)
            {hard (rel. err.)};
        \path[ourwork] (10.05,1.87) rectangle (10.38,2.07);
        \node[anchor=west,font=\scriptsize] at (10.50,1.97)
            {easy (rel. err.)};
        \path[coverage] (12.80,1.87) rectangle (13.13,2.07);
        \node[anchor=west,font=\scriptsize] at (13.25,1.97)
            {easy (TVD err.)};

        \draw[axis,-{Latex[length=2.2mm]}] (0,0) -- (13.30,0)
            node[right, font=\small] {circuit depth $d$};
        \foreach \x in {2.10,2.90,8.50,10.80}
            \draw[tick] (\x,-0.10) -- (\x,0.10);

        \node[depthlabel,anchor=north east] at (2.06,-0.14)
            {$\Theta(p^{-1})$};
        \node[depthlabel,anchor=north west] at (2.94,-0.14)
            {$\Theta\!\left(p^{-1}\log p^{-1}\right)$};
        \node[depthlabel,below=3pt] at (8.50,0)
            {$\Theta\!\left(p^{-1}\log n\right)$};
        \node[depthlabel,below=3pt] at (10.80,0)
            {$\omega(\log n)$};
    \end{tikzpicture}
    \caption{Known complexity regimes as a function of noisy
    circuit depth.  Prior work appears in the top row and our result is in the
    lower row.  Light red denotes relative-error hardness, light blue denotes
    total-variation-distance classical simulability, and solid blue denotes
    relative-error classical simulability.}
    \label{fig:depth-regimes}
\end{figure}
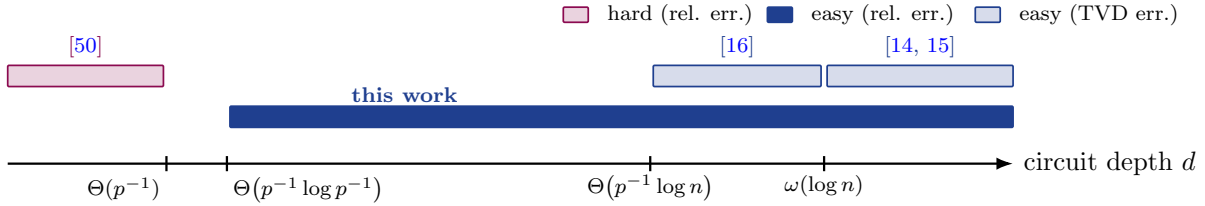

Our sampling algorithm is novel and combines cluster expansion techniques with Pauli path techniques. 
%
The Pauli path framework is typically used in classically simulations of (often random) quantum circuits \cite{schuster_polynomial,garcia_cirac_trivedi_pauli,fontana,angrisani_mele_simulating,mele_noise-induced,martinez_simulation,stilck-franca_garcia-patron,de-palma_marvian_rouze_stilck-franca_optimal_transport,bremner_achieving,gao_duan,aharonov_polynomial_2023,schuster_polynomial,angrisani_schmidhuber_rudolph_cerezo_holmes_huang}, while
the cluster expansion framework is an established tool for rigorously controlling statistical-mechanical systems perturbed away from exactly solvable points \cite{CE_Mayer1937,kotecky_preiss_1986,
netocny_redig_2004,friedli_velenik}, 
which has gained tremendous interest in quantum systems including high-temperature or weakly interacting systems~\cite{
harrow_classical_algorithms_correlation,
mann_efficient_algorithms_approximating,
mann_minko_2024,
chen_rouze_chen_fermions_2025,bakshi_high-temperature,yin_polynomial-time_classical_sampling,zlokapa_quasipolynomial_syk_2026,zlokapa_syk,midha2026beyond,midha2026algorithmic,CE_shorttimeLRB23,CE_metastab25,CE_ldpcgap25}. We develop a novel mapping between marginals of the noisy quantum circuit and partition functions of \textit{abstract polymer models}, allowing us to apply the two in conjunction. Our classical algorithm groups together spatially adjacent `Pauli paths' in the quantum circuit together, and interprets them as \textit{polymers}. We then compute a truncated cluster expansion of the log-marginal, interpreted as the log-partition function of the abstract polymer model, using a combination of Pauli path and tensor network techniques. We prove that this becomes efficient and accurate when the circuit depth exceeds a constant threshold.

Our analysis is inspired by prior work \cite{zhang_lee_approx_markov} which introduces cluster expansion as an alternative route to prove that classical simulability occurs at high \textit{noise strength} greater than a critical threshold. We introduce a method to account for the local accumulation of noise via Pauli paths, within a similar framework of cluster expansion.  A key analytical tool we use is quantum hypercontractivity for depolarizing channels \cite{king_hypercontractivity} and $\ell_2$-norm contraction bounds. At a high level, hypercontractivity allows us to use a few layers of noise to incur a contraction of $2^{-O(n)}$ in Schatten-norms. To our knowledge, the application of this tool in the context of classical simulation is relatively new (e.g. see recent work \cite{shao_zhao_cheng_liu_tensor_network}), and we suspect that this will have broader applications. Finally, our results also build off a line of work \cite{rajakumar_watson_liu,nelson_rajakumar_naturally,oh_classical_simulability_linear_optical,park_oh_mps_noisy_iqp} which first exhibited this constant depth phase transition in noisy quantum circuits with non-universal gate sets. These classical algorithms sample from the output distribution without computing marginals by using randomness in the noise, i.e. simulating noise channels by stochastically sampling errors, and inducing a `percolation' effect in the circuit. In our work, we obtain a deterministic FPTAS for marginals (and can handle universal gate sets) by instead using the contractive effects of noise.

\begin{table}[t] 
\centering\small
\renewcommand{\arraystretch}{1.1}
\begin{tabular}{@{}lllll@{}}
\toprule
\textbf{Work} & \textbf{Circuit family} & \textbf{Depth} & \textbf{Complexity of Sampling} \\
\midrule
\cite{aharonov_ben-or} & geo.\ local, reset gates & any & \hard{100}{Hard (with TVD error)} \\
\cite{ben-or_gottesman_refrigerator,shtanko_sharma_nonunital_1D} & geo.\ local, biased noise & any & \hard{100}{Hard (with TVD error)} \\
\cite{aharonov_limitations,anshu_liu_nguyen_pattison_loglogn_iqp} & non-local & $O(\log n),\geq O(\log\log n)$ & \hard{100}{Hard (with TVD error)} \\
\cite{bravyi_gosset_koenig_tomamichel} & geo. local & $ \leq \Omega(p^{-1})$ & \hard{50}{Hard (for shallow circuits)}\\
\cite{fujii_tamate} & geo.\ local & $\leq \Omega(p^{-1})$ & \hard{50}{Hard (with relative error)} \\
\midrule
This work: & geo.\ local & $\geq O(p^{-1}\log p^{-1})$ & \easy{100}{Easy (with relative error)} \\
\midrule
\cite{rajakumar_watson_liu,nelson_rajakumar_naturally,oh_classical_simulability_linear_optical,park_oh_mps_noisy_iqp} & restricted gate sets & $\geq O(p^{-1}\log p^{-1})$ & \easy{50}{Easy (with TVD error)} \\
\cite{bremner_achieving,gao_duan,aharonov_polynomial_2023,schuster_polynomial} & random gates or inputs & $\ge$ anticoncentration depth & \easy{50}{Easy (with TVD error)} \\
\cite{nelson_rajakumar_limitations} & geo.\ local & $\geq O(p^{-1}\log n)$ & \easy{50}{Easy (with TVD error)} \\
\cite{aharonov_limitations} & any & $\omega(p^{-1}\log n)$ & \easy{50}{Easy (with TVD error)} \\
\bottomrule
\end{tabular}
\caption{Landscape of noisy quantum circuit families in different depth regimes, ordered in roughly decreasing order of computational power. Color thickness encodes strength of the result. Hardness means an efficient classical sampler would imply collapse of the polynomial hierarchy, sometimes under additional conjectures. Our result holds for arbitrary geometrically local circuits, so it subsumes/improves the restricted-gate set and $\Omega(\log n)$-depth simulability results below it.}
\label{table:1}
\end{table}

\subsection{Intuition for Result} \label{sec:intuition}
Why should one expect noisy geometrically local quantum circuits to become classically simulable above a constant depth threshold? Here, we provide an intuitive explanation of the phenomenon and show how cluster expansion arises as a natural proof technique. As a toy example, suppose that all interspersed noise can be shifted to the input product state. In this case, the noisy input state is exactly the thermal state of a 1-local Hamiltonian, where \textit{high depth of noise layers corresponds to high temperature}. Then, the unitary gates in the quantum circuit simply conjugate the state's parent Hamiltonian, keeping locality bounded by the `lightcone,' and at temperatures above a threshold, it is efficiently classically sampleable via cluster expansions \cite{yin_polynomial-time_classical_sampling,bakshi_high-temperature,putterman_zlokapa_cotler_thermal_states}. Note that this example is a significant simplification of how noise behaves in most circuits. However, it still suggests that truncated cluster expansions may be the right classical algorithm. In our proof, we map the output of a noisy quantum circuit to a different statistical-physics model than the thermal distribution of a local Hamiltonian (i.e. abstract polymer models), and we show that a similar high-temperature-like regime is entered after a critical depth threshold, enabling cluster expansion to be successful.

We also point out that geometric locality is fundamentally necessary for the result. Non-local connectivity can allow the circuit to transfer noise to other qubits faster than it accumulates, potentially preserving a subsystem of pure qubits for use in long computations \cite{aharonov_limitations}. However, geometrical locality restricts how fast the gates can move noise out of a subsystem. At each time step, the amount of entropy that is introduced into a subsystem from the noise is roughly $p$ times its volume whereas the amount that can escape is constrained by its boundary, which grows slower than volume in geometrically local systems. This discrepancy enables the high-temperature-like phase to set in after constant depth.


\subsection{Proof Strategy} \label{sec:techniques}
 
At a high level, we will be classically simulating the Heisenberg evolution of the desired observable (i.e. a projector onto a marginal) backwards through the circuit. Our first step is to decompose this observable in the \textit{Pauli basis}, coarse-grain the lattice into `blocks' of side length $2d$, and group the Pauli observables based on which blocks they have support on. The side length of $2d$ ensures that the expectation value of any observable supported only on non-adjacent blocks factorizes into the \textit{product} of expectation values of observables supported on each block. This is because the geometrically local circuit does not have enough depth to cause interactions between two non-adjacent blocks during the Heisenberg evolution. Note this also holds for observables supported on non-adjacent `connected components' in the adjacency graph of blocks.

The key observation is that a similar product structure among connected components on a lattice has been studied before and exploited by classical algorithms in the context of \textit{abstract polymer models}. Informally, a polymer model is defined by a set of weighted objects, called polymers, and an incompatibility relation among them that
describes which polymers factorize. The canonical task is to sample sets of pairwise compatible polymers with probability
proportional to the \textit{product} of their weights, and this often requires computing the normalization of the distribution, i.e. a partition function. In our setting, we identify a set of polymers with connected components of blocks and we say two polymers are incompatible if they overlap or are adjacent. We then carefully define the weight of a polymer $\gamma_i$ as a sum of Pauli observables on $\rho$ that are supported only on $\gamma_i$, so that the partition function of this abstract polymer model is exactly the target marginal probability. This mapping allows us to import over many existing algorithms and analysis tools used for estimating partition functions. It also means that the weight of small polymers can be computed efficiently since only a portion of the circuit needs to be classically simulated (e.g. by brute force).

Finally, we apply cluster expansion, which expands the \textit{logarithm} of the partition function as a sum of products of the weights of various polymers arranged in `clusters'. We show that the noise causes these cluster contributions to decay exponentially in the number of blocks the polymers are supported on, which also ensures that this cluster expansion converges. This is proved via a chain of hypercontractive and contractive inequalities for depolarizing channels that take advantage of geometric locality to gain independent decay factors from different blocks. The classical algorithm simply enumerates all clusters of polymers that are supported on $O(\log n)$ blocks and sums up their weights. This produces an additive approximation to the logarithm of the marginal.

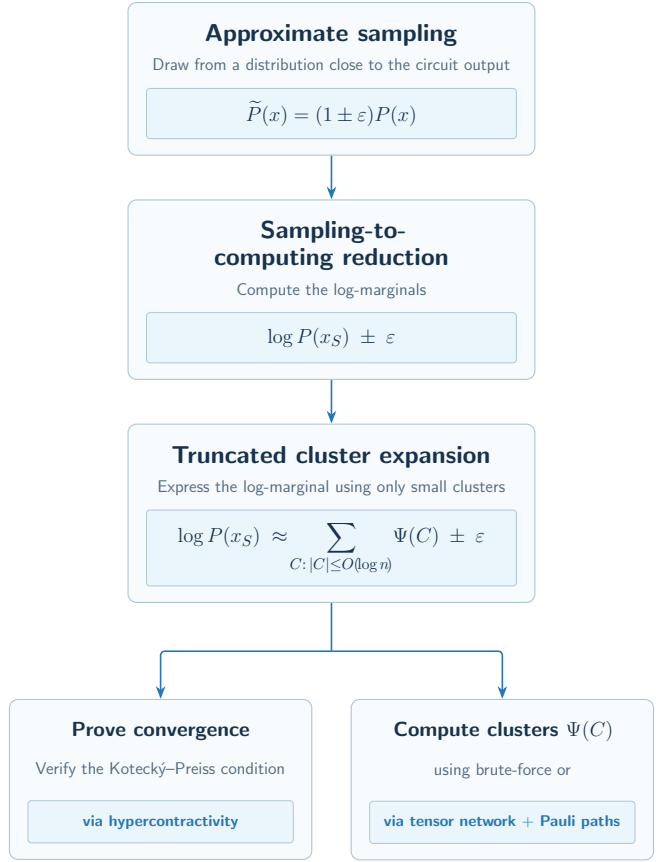
\begin{figure}[ht]
    \centering
    \resizebox{0.55\linewidth}{!}{\input{figures/proof-strategy-flowchart.tikz}}
    \caption{Overview of the proof strategy. The sampling problem reduces to
    computing log-marginals (\cref{sec:problem-statement}, \cref{app:omitted}), which are
    approximated by a truncated cluster expansion
    (\cref{sec:cluster-expansion-overview}). Hypercontractivity proves
    convergence (\cref{sec:hypercontractivity}), while tensor-network and
    Pauli-path methods compute the individual cluster contributions
    (\cref{app:pauli-paths}).}
    \label{fig:proof-strategy}
\end{figure}

\subsection{Open questions}\label{sec:open-questions}
We suspect that the output state of a noisy geometrically local quantum circuit at \textit{any} depth is approximately the thermal state of a local Hamiltonian with temperature that grows with $d$ and $p$. This would offer a more conceptually direct way to prove that classical simulability occurs after a critical circuit depth, however it would naively require studying a cluster expansion over non-commuting terms. Without noise, the output state is the ground state of a local Hamiltonian whose interaction range can grow with circuit depth, and the intuition is that interspersed noise may translate to thermal-like fluctuations on this state without blowing up the locality of the Hamiltonian significantly. It is possible to show that our approximation to the \textit{dephased} output state with output probabilities along the diagonal has a parent Hamiltonian with local terms, beyond the critical depth threshold (though we omit this formulation for clarity). We would like to understand the parent hamiltonian of the non-dephased state, at any depth and noise strength. This may also allow us to make the phase transitions more fine-grained, e.g. when is there separability/death of entanglement \cite{bakshi_high-temperature,putterman_zlokapa_cotler_thermal_states,scalet_entanglement_sudden_death,bakshi_spin_chains}?

Next, hardness of sampling with low \textit{total‑variation distance} requires strong but plausible average-case conjectures \cite{aaronson_arkhipov, bremner_average, bouland_fefferman_nirkhe_vazirani_complexity_rcs} that are believed to hold for certain \textit{noiseless} quantum circuits. Assuming such a conjecture, are there \textit{noisy} geometrically local quantum circuits which hard to sample with low total-variation distance (e.g. at constant depth)? There are a few reasons to be optimistic: (1) noiseless cluster states can encode deep computations that are thought to be hard to approximately sample \cite{hangleiter_anticoncentration,bermejo-vega_hangleiter_architectures}, (2) cluster states display robustness to noise in certain non-adaptive settings \cite{raussendorf_bravyi_harrington_cluster_states,bergamaschi_liu_single_shot,bravyi_gosset_koenig_tomamichel}, and (3) relative-error sampling from noisy cluster states is classically hard assuming the polynomial hierarchy does not collapse \cite{fujii_tamate}. However, there are a few subtleties: (1) the cluster states from \cite{hangleiter_anticoncentration,bermejo-vega_hangleiter_architectures} become easy under noise, (2) existing non-adaptive fault tolerance protocols for cluster states \cite{raussendorf_bravyi_harrington_cluster_states,bergamaschi_liu_single_shot,bravyi_gosset_koenig_tomamichel} cannot handle magic gates/states, and (3) worst-to-average reductions typically used to boost relative-error sampling hardness to approx. sampling hardness fail in constant-depth \cite{napp}, so are inapplicable to \cite{fujii_tamate}. We also remark that an easier problem is open --- noisy quantum circuits at constant depth with \textit{all-to-all connectivity} whose output distributions are hard-to-approximately-sample \cite{rajakumar_watson_liu,anshu_liu_nguyen_pattison_loglogn_iqp}, and an even easier problem is open --- Hamiltonians with constant locality and degree, whose constant-temperature thermal states are hard-to-approximately-sample\footnote{Given a constant-depth quantum cicuit $U$ with hardness that is robust to noise, the Hamiltonian $H =  \sum_i U\ketbra{0}_i U^\dagger$ would solve this open problem.} \cite{rajakumar_watson,bergamaschi_chen_liu}.

Finally, a recent line of work \cite{lee_ghosh_approx_markov,watts_gosset_liu_soleimanifar,zhang_lee_approx_markov,wei_nelwon_rajakumar_2d_clifford,wei_rajakumar_nelson_noise-induced_trunctation,liu_mcginley_schuster_gosset_scrooge} posits the classical simulability of noisy \textit{random} quantum circuits at \textit{any circuit depth} via a conjectured and numerically supported `approximate markov property' of the output distribution. Our convergent cluster expansion can be shown to imply such a property (albeit in the high-depth setting), and our techniques may provide a stepping stone towards the random setting. We remark however that any application towards random quantum circuits will need to intrinsically take advantage of the effects of disorder, and it is unclear how to adapt our `worst-case' techniques in such a direction. Relatedly, it would be interesting to extend our method to other types of noise, such as non-unital noise \cite{nonunital_fefferman_ghosh_soumik_gullans,shtanko_sharma_nonunital_1D}, for which the cluster expansion may be developed around fixed points other than the maximally mixed state.


\subsection{Outline}

The rest of the paper is organized as follows. We start with a formal description of the computational task and theorem statement in \cref{sec:problem-statement}. In \cref{sec:polymer-expansion}, we describe the exact mapping between a marginal and a partition function of a polymer model. In \cref{sec:cluster-expansion-overview}, we explain how a classical algorithm based on cluster expansion can estimate the logarithm of this partition function accurately, assuming a `weight decay' inequality. In \cref{sec:hypercontractivity}, we prove this inequality holds for noisy quantum circuits with depth greater than a critical threshold. Sections \ref{sec:cluster-expansion-overview} and \ref{sec:hypercontractivity} together establish that the algorithm is accurate, but its runtime is superpolynomial when $d$ grows with $n$. In Appendix \ref{app:pauli-paths}, we explain how a combination of `Pauli path' techniques and tensor networks brings the runtime down to polynomial in all depth regimes, while retaining accuracy. We also more formally specify the full polynomial-time algorithm in \cref{app:pauli-paths}.

\section{Problem statement}\label{sec:problem-statement}

We consider a quantum circuit on $n$ qubits with depth $d$. For each time step $t\in[d]$, let $\mathcal U^{(t)}$
be a tensor product of two-qubit unital CPTP maps
acting on disjoint nearest-neighbor pairs, with
identity channels on unmatched qubits.

We assume that the circuit is \emph{geometrically local}: the qubits occupy a
$D$-dimensional lattice, and every two-qubit operation acts on a pair of nearest
neighbors.  Here $D$ denotes the spatial dimension, which we take to be a
constant, $D=O(1)$.

For each qubit $i\in[n]$, define the single-qubit depolarizing channel of strength $p\in(0,1)$ by
\begin{equation}\label{eq:single-qubit-depolarizing-channel}
    \mc N_i(\rho)
    :=(1-p)\rho
      +p\frac{\mathbb I_i}{2}\otimes\Tr_i(\rho).
\end{equation}
We write
\begin{equation}
    \mc N_{[n]}:=\bigcirc_{i\in[n]}\mc N_i
\end{equation}
for one layer of independent depolarizing noise acting on all qubits.

\begin{definition}[Noisy quantum circuit]\label{def:noisy-quantum-circuit}
The noisy implementation of the depth-$d$ circuit is the channel
\begin{equation}\label{eq:noisy-circuit}
    \Phi
    :=\left(\mc N_{[n]}\circ\mc U^{(d)}\right)
      \circ\cdots\circ
      \left(\mc N_{[n]}\circ\mc U^{(1)}\right).
\end{equation}
The input is the all-zero state
\begin{equation}\label{eq:zero-input-state}
    \rho_0:=\ketbra{0}{0}^{\otimes n},
\end{equation}
and the output state is
\begin{equation}\label{eq:noisy-output-state}
    \rho_{\mathrm{out}}:=\Phi(\rho_0).
\end{equation}
\end{definition}

At the end of the circuit, all qubits are measured in the computational basis. The resulting output distribution $P_\Phi$ on $\{0,1\}^n$ is
\begin{equation}\label{eq:noisy-output-distribution}
    P_\Phi(x)
    :=\bra{x}\rho_{\mathrm{out}}\ket{x}
    =\bra{x}\Phi(\rho_0)\ket{x},
    \qquad x\in\{0,1\}^n.
\end{equation}




In order to carry out this sampling task, we use the well-known sampling-to-computation reduction based on sequentially
computing marginal probabilities (see \cite{bremner_achieving,yin_polynomial-time_classical_sampling}).

For a subset $S\subseteq[n]$ and a partial outcome
$x_S\in\{0,1\}^{S}$, define the marginal probability
\begin{equation}\label{eq:output-marginal}
    P_\Phi(x_S)
    :=\Tr\!\left[
        \left(\ketbra{x_S}{x_S}\otimes\mathbb I_{\bar S}\right)
        \Phi(\rho_0)
      \right].
\end{equation}
Because the final layer of depolarizing noise has strength $p>0$, these
marginals are strictly positive, so their logarithms are well defined.

\begin{proposition}[Sampling-to-Computing, proved in Appendix \ref{app:omitted}]
\label{prop:sampling-from-log-marginals}
Suppose that, for every $S\subseteq[n]$ and
$x_S\in\{0,1\}^{S}$, one can compute $\log P_\Phi(x_S)$ in time
$T$ to additive error at most $\epsilon \leq 1$. Then one can sample in time
$O(nT)$ from a distribution $\widetilde{P}_\Phi$ satisfying a multiplicative error bound,
\begin{align} \label{eq:multiplicative}
    \forall x\in\{0,1\}^n, e^{-2n\epsilon } P_\Phi(x)  \leq \widetilde{P}_\Phi(x) \leq e^{2n\epsilon}P_\Phi(x)
\end{align}
as well as a total variation distance bound,
\begin{align} \label{eq:tvd}
    \sum_{x\in\{0,1\}^n}\frac{1}{2}|P_\Phi(x)-\widetilde{P}_\Phi(x)| \leq \frac{1}{2}(e^{2n\epsilon }-1)
\end{align}
\end{proposition}

Consequently, pointwise relative error $\epsilon$
is obtained by computing each log-marginal to
additive error $\log(1+\epsilon)/(2n)$. Thus, we can state our main theorem,

\begin{theorem} [FPTAS for arbitrary marginals]\label{thm:formal}
    Let $\Phi$ be a noisy quantum circuit as in Definition \ref{def:noisy-quantum-circuit}, such that $d > d_{crit} = O(\frac{D}{p}\log \frac{D}{p} + \frac{3^D}{p})$. Given an error budget $0 < \epsilon \leq 1$, a deterministic algorithm computes $\eta$ satisfying $|\eta - \log P_{\Phi}(x_S)| \leq \epsilon $ for every $S\subseteq[n]$ and
$x_S\in\{0,1\}^{S}$, in runtime $\min((n/\epsilon)^{O((2d)^D)},(n/\epsilon)^{O(\frac{D^{D}}{p^{D+1}} )})$
\end{theorem}

\section{Polymer model for quantum circuits}\label{sec:polymer-expansion}

Our goal is now to calculate $\log P_\Phi(x_S)$ to additive error. To do so, we first write the marginal probability in a form reminiscent of a partition function by expanding
each computational-basis projector in the Pauli basis.
\begin{align}
\label{eq:pauli-partition-function}
\Xi := 2^{|S|} P_\Phi(x_S)
    &= \Tr\!\left[
        \left(
          \bigotimes_{i\in S}
          \left(\mathbb I_i+(-1)^{x_i}Z_i\right)
        \right)
        \Phi(\rho_0)
      \right] \\
    &= \sum_{A\subseteq S}
       \Tr\!\left[Z_A\Phi(\rho_0)\right],
       \label{eq:pauli-sum}
\end{align}
where
\begin{equation}
    Z_A:=\prod_{i\in A}(-1)^{x_i}Z_i,
    \qquad Z_\varnothing:=\mathbb I,
\end{equation}
with each $Z_i$ acting as the Pauli $Z$ operator on qubit $i$ and as the identity elsewhere.

In statistical mechanics, partition functions can often be computed in the
abstract framework of polymer models (e.g. see \cite[Chapter~5.2]{friedli_velenik}). In this framework, the main strategy is to rewrite $\Xi$ into a sum of factorizable terms. In particular, $\Xi$ is written in terms of a set of polymers which we denote as $\Gamma$. The contribution of a given polymer $\gamma \in \Gamma$ is called its weight and denoted by $w(\gamma)$. We say that two polymers are compatible if their contribution to $\Xi$ factorizes. More formally, a polymer expansion is a set of polymers $\Gamma$ and weight function $w$ such that
\begin{align}
    \Xi = \sum_{\text{compatible } \Gamma' \subseteq \Gamma} \prod_{\gamma \in \Gamma'} w(\gamma)
    \label{eq:abstract-polymer-partition-function}
\end{align}

In order to arrive at the desired polymer expansion of \cref{eq:abstract-polymer-partition-function}, we now reorganize the sum in \cref{eq:pauli-sum} according to
the geometry of the circuit. We first partition the physical lattice into coarse-grained
blocks of side length $2d$ in each spatial direction, and let $\mathcal J$
be the set of blocks that intersect $S$. We consider blocks to be adjacent when they share a face, edge or corner.
Importantly, we have chosen the block size exactly large enough such that any two non-adjacent blocks have disjoint backwards lightcones. This will be crucial when factorizing polymer weights.

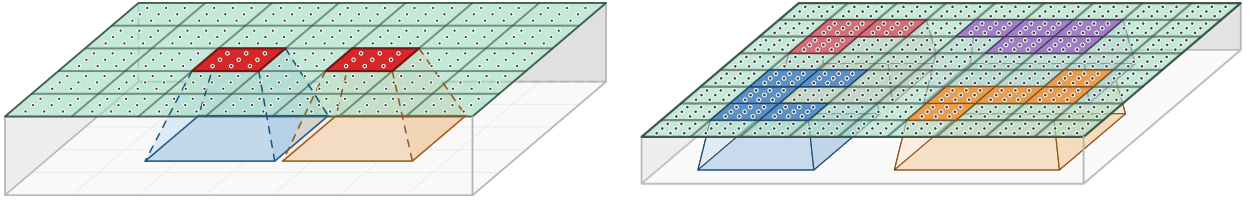
\begin{figure}
    \centering
    \begin{subfigure}[t]{0.49\linewidth}
        \vspace{0pt}
        \centering
        \begin{minipage}[t][3.25cm][t]{\linewidth}
            \vspace{0pt}
            \centering
            \resizebox{\linewidth}{!}{%
                \input{figures/backward-lightcone.tikz}%
            }
        \end{minipage}
        \caption{Backward light cones of two non-adjacent blocks.}
        \label{fig:disjoint-backward-lightcones-panel}
    \end{subfigure}\hfill
    \begin{subfigure}[t]{0.49\linewidth}
        \vspace{0pt}
        \centering
        \begin{minipage}[t][3.25cm][t]{\linewidth}
            \vspace{0pt}
            \centering
            \resizebox{\linewidth}{!}{%
                \input{figures/lattice-animals.tikz}%
            }
        \end{minipage}
        \caption{Compatible polymers and their backward light-cones.}
        \label{fig:lattice-animals}
    \end{subfigure}
    \caption{The coarse-grained geometry. (a) shows that the backward light-cones of any two non-adjacent blocks are disjoint. (b) depicts a collection of compatible polymers which are disjoint connected sets of blocks. The dots denote physical
    qubits within each coarse-grained block.}
    \label{fig:disjoint-backward-lightcones}
\end{figure}

For $A\subseteq S$, define its coarse support $\operatorname{Supp}_{\mathrm{cg}}(A) \subseteq \mathcal{J}$ to be the set of coarse-grained blocks containing at least one qubit of $A$. With this coarse-graining in place we can now rewrite $\Xi$ in \cref{eq:pauli-sum} as:

\begin{align}
    \Xi
      =\sum_{J\subseteq\mathcal J}
        \Tr\!\left[ \left(
          \sum_{\substack{A\subseteq S:\\
          \operatorname{Supp}_{\mathrm{cg}}(A)=J}}\Phi^\dagger(
  Z_A)\right)
          \rho_0
      \right]
      \label{eq:coarse-support-decomposition}
\end{align}

We can now leverage the fact that non-adjacent blocks are lightcone separated. In particular, the connected components
of $J$, which we denote by $\text{cc}(J)$, are non-adjacent by definition and therefore have disjoint light cones. We first rewrite our observable as a product of disjoint observables each acting on a separate connected component's lightcone.

\begin{align}
 \sum_{\substack{A\subseteq S:\\
          \operatorname{Supp}_{\mathrm{cg}}(A)=J}}\Phi^\dagger(
  Z_A) = \prod_{\gamma \in \text{cc}(J)} \Phi^\dagger\left(
  \sum_{\substack{A\subseteq S:\\
          \operatorname{Supp}_{\mathrm{cg}}(A)=\gamma}}
  Z_A
\right) = \prod_{\gamma \in \text{cc}(J)} \Phi^\dagger(M_\gamma)
\end{align}
where,
\begin{align}
    M_\gamma := \sum_{\substack{A\subseteq S:\ \,
\operatorname{Supp}_{\mathrm{cg}}(A)=\gamma}} Z_A.
\label{eq:mgamma}
\end{align}

Since each $\Phi^\dagger(M_\gamma)$ has disjoint support and $\rho_0$ is a product state, we have
\begin{align}
    \Tr\!\left[\left(\prod_{\gamma \in \text{cc}(J)} \Phi^\dagger(M_\gamma)
\right)
          \rho_0
      \right] &=
      \prod_{\gamma\in\text{cc}(J)}
      \Tr\!\left[
        \Phi^\dagger(M_\gamma)\rho_0
      \right]
\end{align}
Thus,
\begin{align}
    \Xi = \sum_{J\subseteq\mathcal J}\prod_{\gamma\in\text{cc}(J)}
      \Tr\!\left[
        \Phi^\dagger(M_\gamma)\rho_0
      \right]
\end{align}
The expression now identifies the natural polymer model. A
\emph{polymer} is a nonempty connected set
$\gamma\subseteq\mathcal J$, and its weight is
\begin{equation}\label{eq:polymer-weight}
    w(\gamma)
    := \Tr\!\left[
        \Phi^\dagger(M_\gamma)\rho_0
      \right].
\end{equation}
Let $\Gamma$ be the set of all such polymers. Two polymers are compatible
when they contain no common block and no block in one is a neighbor of a
block in the other. The connected
components of each $J\subseteq\mathcal J$ form a compatible collection,
and every compatible collection arises uniquely in this way. Hence $\Xi$
has precisely the abstract polymer expansion given in
\cref{eq:abstract-polymer-partition-function}. The choice
$J=\varnothing$ corresponds to the empty polymer collection and
contributes $1$.

\section{Cluster expansion}\label{sec:cluster-expansion-overview}

Now that we have written $\Xi$ as a sum over compatible polymers we can attempt to perform a cluster expansion, which expands $\log \Xi$ as a sum of exponentially decaying cluster contributions. We first provide a brief overview of this method following \cite{friedli_velenik}.

We write $\delta(\gamma,\gamma')=1$ when $\gamma$ and $\gamma'$ are
compatible and $\delta(\gamma,\gamma')=0$ when they are incompatible. In
particular, $\delta(\gamma,\gamma)=0$. Next we define the Ursell function:
\begin{equation}
    \varphi(\gamma_1,\ldots,\gamma_k)
    :=\frac{1}{k!}
      \sum_{\substack{\text{connected graph }G_k\\
                      \text{with }k\text{ vertices }\{1,\ldots,k\}}}
      \prod_{\substack{\text{edge }\{i,j\}\\
                       \text{in }G_k}}
      \bigl(\delta(\gamma_i,\gamma_j)-1\bigr),
      \qquad k>1.
\end{equation}
and $\varphi(\gamma_1)=1$. We will not use this expression explicitly; it is sufficient for our purpose to view $\varphi$ as some function
determined by the compatibility relation.

A multiset of polymers $C=\{\gamma_1, \dots, \gamma_k\}$ is called a \emph{cluster} if it cannot be
separated into two nonempty sets $C=\Gamma_1\sqcup\Gamma_2$ such that every
$\gamma_1\in\Gamma_1$ is compatible with every
$\gamma_2\in\Gamma_2$. Equivalently, form a graph whose vertices are the
polymers in $C$ and join two polymers when they are incompatible; then $C$
is a cluster precisely when this graph is connected. Intuitively, two polymers are incompatible if they touch in space, and so a cluster is a ``connected'' set of polymers. Note that since each cluster is multiset, the same polymer can appear multiple times.

The following
cluster expansion result states that the
logarithm of the partition function can be written as a sum over clusters:

\begin{proposition}[Theorem~5.4 and Eq.~(5.20) of \cite{friedli_velenik}]\label{prop:cluster-expansion}
Suppose there exists a function $a:\Gamma\to[0,\infty)$ such that
\begin{equation}\label{eq:cluster-convergence-condition}
    \sum_{\gamma\text{ incompatible with }\gamma'}
      |w(\gamma)|e^{a(\gamma)}
    \leq a(\gamma')
    \qquad\text{for every }\gamma'\in\Gamma.
\end{equation}
Then
\begin{equation}\label{eq:log-xi-cluster-expansion}
    \log\Xi
    =\sum_{k\geq1}
      \sum_{\gamma_1}\cdots\sum_{\gamma_k}
      \varphi(\gamma_1,\ldots,\gamma_k)
      \prod_{i=1}^{k}w(\gamma_i)
    =\sum_{\text{cluster }C}\Psi(C),
\end{equation}
where the contribution of a cluster $C=\{\gamma_1,\ldots,\gamma_k\}$, is
\begin{equation}\label{eq:cluster-contribution}
    \Psi(C)
    :=\frac{k!}{\displaystyle\prod_{\gamma\in\Gamma}n_C(\gamma)!}
      \varphi(\gamma_1,\ldots,\gamma_k)
      \prod_{i=1}^{k}w(\gamma_i).
\end{equation}
where $n_C(\gamma)$ denotes the number of times the polymer $\gamma\in\Gamma$ appears
in $C$.
Furthermore, for every $\gamma_1\in\Gamma$,
\begin{equation}\label{eq:ursell-summability-bound}
    1+\sum_{k\geq2}k
      \sum_{\gamma_2,\ldots,\gamma_k}
      |\varphi(\gamma_1,\ldots,\gamma_k)|
      \prod_{j=2}^{k}|w(\gamma_j)|
    \leq e^{a(\gamma_1)}.
\end{equation}
\end{proposition}

For our model, we are able to show that the polymer weights decay exponentially in their size, which is later proven in \cref{sec:hypercontractivity}, when circuit depth exceeds the critical threshold. More formally, we show that:
\begin{align}\label{eq:polymer-weight-decay}
|w(\gamma)| \leq e^{-c|\gamma|},
\end{align}
where $|\gamma|$ here is taken to be the number of blocks in the connected set $\gamma$ and $c>0$ is any constant that can be chosen later (affecting the critical threshold by a constant factor). Combining this with \cref{prop:cluster-expansion} implies that ``large'' clusters have exponentially small contribution.
To see this, first consider the modified weight function
\begin{align} \label{eq:modified-weight}
w'(\gamma) := w(\gamma) e^{c'|\gamma|}
\end{align}
for some constant $c'$ that can be chosen such that $c>c'>0$. This satisfies \cref{eq:cluster-convergence-condition} by the following calculation. Letting $a(\gamma) = |\gamma|$,

\begin{align}
\sum_{\gamma\text{ incompatible with }\gamma'}
      |w'(\gamma)|e^{|\gamma|}
    &\leq \sum_{m=1}^\infty |\gamma'|(e3^D)^{m}e^{(-c+c'+1)m}
      \label{eq:incompatible-polymer-count-bound} \\
    &\leq |\gamma'|,
      \qquad \text{provided that }c-c'\geq 1+\log\!\left(2e3^D\right).
\end{align}
The bound in \cref{eq:incompatible-polymer-count-bound} follows because the number of size-$m$ polymers incompatible with $\gamma'$ is at most $|\gamma'|(e3^D)^{m}$ by a standard counting bound \cite[Lemma~2]{beveridge_frieze_mcdiarmid}. In order to obtain a bound on the total contribution of all clusters we introduce a couple more notions. We say that $\gamma$ \emph{contains} a block $x$ when $\gamma\ni x$, and that a cluster $C$ contains $x$ when at least one polymer in $C$ contains $x$. We define the size of a cluster by $|C|:=\sum_{\gamma\in C}|\gamma|$. We can now bound a weighted sum of cluster coefficients as:
\begin{equation}\label{eq:all-cluster-summability}
\begin{aligned}
\sum_{\text{cluster }C}|\Psi(C)|e^{c'|C|}
&\leq \sum_{x\in\mathcal J}\sum_{k\geq1}
    \sum_{\substack{\gamma_1\ni x\\ \gamma_2,\ldots,\gamma_k}}
    k\,|\varphi(\gamma_1,\ldots,\gamma_k)|
    \prod_{i=1}^{k}|w(\gamma_i)|e^{c'|\gamma_i|} \\
&= \sum_{x\in\mathcal J}\sum_{\gamma_1\ni x}|w'(\gamma_1)|
    \left[
      1+\sum_{k\geq2}k\sum_{\gamma_2,\ldots,\gamma_k}
      |\varphi(\gamma_1,\ldots,\gamma_k)|
      \prod_{i=2}^{k}|w'(\gamma_i)|
    \right] \\
&\leq \sum_{x\in\mathcal J}\sum_{\gamma_1\ni x}
    |w'(\gamma_1)|e^{a(\gamma_1)}\\
&\leq |\mathcal J|.
\end{aligned}
\end{equation}
The factor $k$ in the second line accounts for the $k$ possible choices of the polymer in the cluster that is fixed to contain $x$, and the final inequality follows from \cref{eq:ursell-summability-bound} applied to the modified weights $w'$. It now follows that,
\begin{align} 
\sum_{C, |C| > m} |\Psi(C)| &\leq e^{-c' m} |\mathcal J| \\
&\leq \epsilon,
\qquad \text{provided that }
m\geq \frac{1}{c'}\log\!\left(\frac{|\mathcal J|}{\epsilon}\right).\label{eq:final-bound}
\end{align}

Thus, we can estimate $\log\Xi$ by the sum
$\sum_{C,|C|\leq m}\Psi(C)$ for
$m=\frac{1}{c'}\log\!\left(\frac{|\mathcal J|}{\epsilon}\right)$
while only incurring additive error at most $\epsilon$.



In Appendix \ref{app:pauli-paths}, we write out the algorithm to compute $\sum_{C,|C|\leq m}\Psi(C)$ explicitly and compute the runtime for each step. The dominant cost in runtime will be the computation of each polymer weight, which can be computed by brute force in time $2^{O(m(2d)^D)}$ since there are $m(2d)^D$ qubits contained in a polymer of size $m$ ($m$ blocks of size $(2d)^D$ each). This results in a total runtime of roughly $|\mathcal J|2^{O(m(2d)^D)}$. This is polynomial in $n$ when $d,D = O(1)$ since $m=O(\log n)$ and quasipolynomial when $d$ grows with $n$ (until $d = O(\log n)$). In Appendix \ref{app:pauli-paths}, we substitute this step with a truncated Pauli path enumeration and tensor network simulation to efficiently approximate polymer weights to a sufficient precision, achieving a polynomial-time algorithm in all depth regimes above the critical threshold.

\section{Weight decay via hypercontractivity}\label{sec:hypercontractivity}
In order to complete the result, it only remains to prove \cref{eq:polymer-weight-decay}. We can do this by first phrasing this bound in terms of the operator norm:
\begin{align}
|w(\gamma)| &=\left|\Tr\!\left[
        \Phi^\dagger(M_\gamma)\rho_0
      \right]\right|
      \leq \norm{\Phi^\dagger(M_\gamma)}_\infty\norm{\rho_0}_1
      = \norm{\Phi^\dagger(M_\gamma)}_\infty.
\end{align}

In order to bound this operator norm we will use a combination of hypercontractivity and contractivity. In particular we will apply the following two results:

\begin{fact}[Hypercontractivity]\label{fact:hypercontractivity}
    Let $f$ be a linear operator acting on the space of $n$ qubits. Let $\mc U_1,\ldots,\mc U_d$ be any sequence of unital CPTP maps. For any $\Delta_1,\Delta_2\in \mathbb{R}_{\geq 0}$ s.t. $(1-p)^{2d} \Delta_1 \leq \Delta_2 \leq \Delta_1$
    \begin{align}
    \norm{\bigcirc_{i \in [d]} [\mc U_i \circ \mc N_{[n]} ](f)}_{1+ \Delta_1} \leq 2^{-n\left(\frac{1}{1+\Delta_2} - \frac{1}{1+\Delta_1}\right)}\norm{f}_{1+\Delta_2}
\end{align}
\end{fact}
\begin{fact}[Contractivity]\label{fact:contractivity}
Let $f$ be a linear operator. Let $\mc U_1,\ldots,\mc U_d$ be any sequence of unital CPTP maps.
\begin{equation}
    \norm{\bigcirc_{i\in[d]}[\mc U_i\circ \mc N_{[n]}](f)}_2
    \le(1-p)^{wd}\norm{f}_2,
\end{equation}
where $w$ is defined by taking the Pauli decomposition before each noise layer and letting $w$ be the minimum weight of any Pauli operator with non-zero coefficient in any of these decompositions. 
\end{fact}
The general proof strategy will be to use hypercontractivity while passing from the $\ell_\infty$ norm to the $\ell_1$ norm. On the way, we stop at the $\ell_2$ norm to use contractivity. And finally, we compute the $\ell_1$ norm exactly.
\begin{align}
    \ell_\infty \leq \ell_{1+(4d)^D} \xrightarrow[\text{hypercontractivity}]{\text{noise layers}} \ell_2 \xrightarrow[\text{contractivity}]{\text{noise layers}} \ell_2 \xrightarrow[\text{hypercontractivity}]{\text{noise layers}} \ell_{1 + \frac{1}{(4d)^D}} \leq \ell_1
\end{align}

Previously, $M_\gamma$ has been implicitly defined as acting on all $n$ qubits even though it only has non-identity support within $\gamma$. Excess identities do not affect the operator norm since $\norm{O\otimes\mathbb I}_\infty=\norm{O}_\infty$; however, they do affect the other Schatten $p$-norms. Therefore, we must trim these identities before performing the hypercontractivity, and so for the rest of this section, $M_\gamma$ is restricted to only have support on the lightcone of $\gamma$.Let $n_{\gamma}$ be the number of qubits within $\gamma$ and its reverse lightcone. Let
$\Phi_{\gamma}^{\dagger}$ denote the channel acting on $n_{\gamma}$ qubits, retaining only those unital operations of
$\Phi^\dagger$ within the light cone of $\gamma$ (but keeping all noise channels). Let $\Phi_{\gamma,1:d}^{\dagger}$ denote the restriction of $\Phi_{\gamma}^{\dagger}$ to only the layers in $[1:d]$ (note in the adjoint channel, noise is applied before the unital operations in each layer).

\begin{lemma}
    $
|w(\gamma)| \leq e^{-c|\gamma|}$ when $d>d_\text{crit} = O(\frac{D}{p} \log \frac{D}{p} + \frac{3^D}{p}c)$.
\end{lemma}

\begin{proof}
Let
$d^*:=\left\lceil\frac{D}{2p}\log 4d \right\rceil$, which ensures that
\begin{equation}\label{eq:d-star-definition}
    (1-p)^{2d^*}(4d)^D\leq 1.
\end{equation}
\begin{align}
\norm{\Phi_{\gamma,1:d}^\dagger(M_\gamma)}_\infty
&\leq \norm{\Phi_{\gamma,1:d}^\dagger(M_\gamma)}_{1+(4d)^D}
&& \text{(monotonicity)}
\notag\\
&\leq
2^{-n_\gamma\left(\frac12-\frac{1}{1+(4d)^D}\right)}
\norm{\Phi_{\gamma,1:d-d^*}^\dagger(M_\gamma)}_2
&& \text{(hypercontractivity)}
\notag\\
&\leq
2^{-n_\gamma\left(\frac12-\frac{1}{1+(4d)^D}\right)}
(1-p)^{\frac{|\gamma|(d-2d^*)}{3^D}}
\norm{\Phi_{\gamma,1:d^*}^\dagger(M_\gamma)}_2
&& \text{(contractivity)}
\notag\\
&\leq
2^{-n_\gamma\left(
    \frac{1}{1+(4d)^{-D}}-\frac{1}{1+(4d)^D}
\right)}
\notag\\[-0.25em]
&\qquad{}\times
(1-p)^{\frac{|\gamma|(d-2d^*)}{3^D}}
\norm{M_\gamma}_{1+(4d)^{-D}}
&& \text{(hypercontractivity)}
\notag\\
&\leq
2^{-n_\gamma\left(1-\frac{2}{1+(4d)^D}\right)}
\notag\\[-0.25em]
&\qquad{}\times
(1-p)^{\frac{|\gamma|(d-2d^*)}{3^D}}
\norm{M_\gamma}_1
&& \text{(monotonicity)}.
\label{eq:weight-decay-hypercontractive-steps}
\end{align}
The contractivity step uses the fact that every Pauli string appearing at an
intermediate time has weight at least $|\gamma|/3^D$. Indeed, one can choose
at least $|\gamma|/3^D$ pairwise non-neighboring blocks of $\gamma$; their
backward light cones are disjoint, so the Pauli string must act nontrivially
on each one. Therefore the total weight during the $d-2d^*$ layers of 2-norm contraction is at least $\frac{|\gamma|(d-2d^*)}{3^D}$.

It remains only to bound the $\ell_1$ norm. The non-identity portion of $M_\gamma$ is contained in the intersection of $\gamma$ and $S$. This is given by
\begin{equation}
\begin{aligned}
\bigotimes_{B\in\gamma}
  \left[
    \bigotimes_{i\in B \cap S}
    \bigl(\mathbb I_i+(-1)^{x_i}Z_i\bigr)-\mathbb I_{B \cap S} 
  \right] = \bigotimes_{B\in\gamma}
  \left[
  2^{|B \cap S|}\ketbra{x}_{B \cap S}-\mathbb I_{B \cap S} 
  \right]
\end{aligned} \label{eq:unevolved-polymer-trace-norm}
\end{equation}
$M_\gamma$ is also padded by identities to extend to the lightcone of $\gamma$ (i.e. it acts on $n_\gamma$ qubits). It is straightforward to compute the $\ell_1$-norm entry-wise using \cref{eq:unevolved-polymer-trace-norm}, and this gives $\norm{M_{\gamma}}_1 \leq 2^{n_\gamma + |\gamma|}$.  Substituting this in  \cref{eq:weight-decay-hypercontractive-steps}
gives
\begin{align}
\norm{\Phi_{\gamma,1:d}^\dagger(M_\gamma)}_\infty
&\leq
2^{\frac{2n_\gamma}{1+(4d)^D}+|\gamma|}
(1-p)^{\frac{|\gamma|(d-2d^*)}{3^D}}
\\
&\leq
2^{3|\gamma|}(1-p)^{\frac{|\gamma|(d-2d^*)}{3^D}}
&& \text{since } n_\gamma\leq (4d)^D|\gamma|
\\
&\leq
\exp\!\left[-|\gamma|\left(
    \frac{p(d-2d^*)}{3^D}-3\log 2
\right)\right] \\
&\leq \exp(-c|\gamma|)
\label{eq:polymer-operator-norm-decay}
\end{align}

It remains to determine how large the circuit depth must be for the last
inequality to hold.
Working through the algebra, a decay rate of $c$ is obtained when
\begin{equation}
    d\geq \frac{2D}{p}\log\!\left(\frac{12D}{p}\right)
      +\frac{2\cdot 3^D}{p}\bigl(c+3\log 2\bigr) + 4.
\end{equation}
Thus, for $d>d_\text{crit} = O(\frac{D}{p} \log \frac{D}{p} + \frac{3^D}{p}c)$, the desired exponential weight decay is achieved.
\end{proof}

\begin{figure}[t]
    \centering
    \resizebox{0.8\linewidth}{!}{\input{figures/contractivity_layers.tikz}}
    \caption{Illustration of norm conversions and applications of contractivity / hypercontractivity at different layers of the circuit.}
    \label{fig:contractivity_layers}
\end{figure}
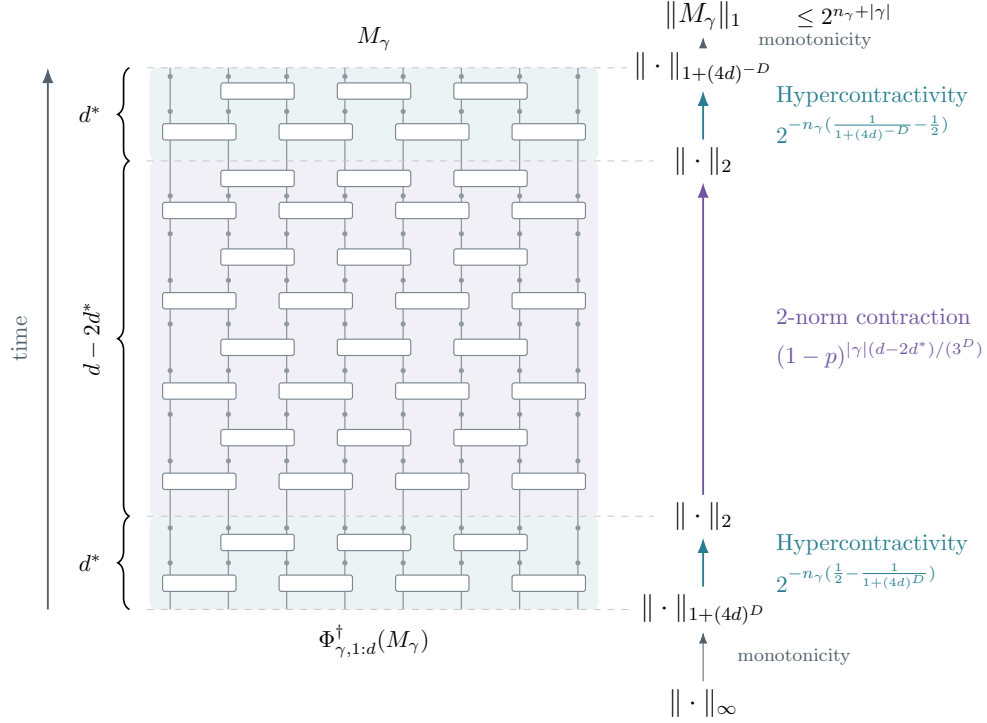

The entire path is illustrated in Figure~\ref{fig:contractivity_layers}. Thus, we obtain our main acccuracy bound by setting $c > 2+\log (2e3^D)$ and following lines \cref{eq:polymer-weight-decay} to \cref{eq:final-bound}, using $c' = 1$ in \cref{eq:modified-weight},
\begin{corollary} \label{corollary:accuracy}
    $|\log\Xi - \sum_{C,|C|\leq m}\Psi(C)| \leq \epsilon$ for
$m=\lceil \log\!\left(|\mc J|/\epsilon\right)\rceil$ when $d>d_\text{crit} = O(\frac{D}{p} \log \frac{D}{p} + \frac{3^D}{p})$.
\end{corollary}

\section{Acknowledgements}
C.Y. thanks Xun Gao for valuable discussions that helped motivate this work. C.Y. is partially supported by the National Science Foundation under Award No.~2016245 and by the Stanford Q-FARM Bloch Postdoctoral Fellowship in Quantum Science and Engineering. This material is based upon work supported by the U.S. Department of Energy, Office of Science, Accelerated Research in Quantum Computing, Fundamental Algorithmic Research toward Quantum Utility (FAR-Qu). J.R. is supported by Department of Navy Award No. N000142612399 issued by the Office of Naval Research. M.J.G acknowledges support from NSF QLCI grant OMA-2120757, NSF QLCI grant OMA-2553574, and NSF NQVL:Design:ORAQL grant 2533041. We thank Yi-Kai Liu, Thiago Bergamaschi, Su-un Lee, Dominik Hangleiter, Zhi-Yuan Wei, Daniel Malz and Alexey Gorshkov for helpful discussions. 

\paragraph{AI Usage Acknowledgement.} The authors have posed and thought about this open problem since 2023. They developed the main reduction to cluster expansion without generative AI tools. AI (specifically ChatGPT Sol 5.6) was used for a specific step of the proof that evaded the authors for several months: the weight decay of \cref{sec:hypercontractivity}. The prompt was the missing lemma statement along with the the cue `I think we need to use hypercontractivity.' Generative AI was also used to flush out the quasipolynomial-to-polynomial time reduction in \cref{app:pauli-paths} and for general reviewing for clarity and errors. The authors have verified the proof themselves and take full responsibility for claims made.

\printbibliography

\appendix

\section{Polynomial time algorithm at all higher depths}\label{app:pauli-paths}

The runtime $|\mathcal J|\,2^{O(m(2d)^D)}$ of \cref{sec:cluster-expansion-overview} comes entirely from computing each polymer weight $w(\gamma)$ by brute force over the light cone of $\gamma$. It turns out that via a combination of Pauli path techniques and tensor network techniques, we can compute a sufficient estimate of $\omega(\gamma)$ in polynomial time for all depths. The main intuition is that: by using brute force simulation during the middle layers where 2-norm contraction occurs, we are keeping track of Pauli paths of weight that may far exceed $\frac{|\gamma|(d-2d^*)}{3^D}$, but we should really be able to get away with truncating a large portion of these paths. The subtlety is that our analysis techniques for proving accuracy required us to perform hypercontractive steps of the initial and final $d^*$ layers, which means we cannot truncate Pauli paths during those portions, and instead need to simulate the entire evolution within them. Since $d^*$ can also scale with $n$ (but much slower than $d$ does), brute force simulation runs into the same quasi-polynomial time issue again, but we can fix it here by using tensor networks. We write out the full algorithm below for completeness, followed by a proof that runtime is polynomial and a proof that the approximation error is bounded. 

\begin{algorithm}[H]
\caption{FPTAS for marginals}
\label{alg:fptas}
\DontPrintSemicolon

\KwIn{
The classical description of a noisy geometrically local circuit $\Phi$ as in \Cref{def:noisy-quantum-circuit}; an error parameter $\epsilon$; a set $S \subseteq [n]$ and bitstring $x \in \{0,1\}^{|S|}$
}
\KwOut{
A number $\eta$ s.t. $|\log P_\Phi(x_S) - \eta| \leq \epsilon$
}

\BlankLine

Set $m \gets \lceil \log \frac{2|\mathcal{J}|}{\epsilon} \rceil$ and $L \gets \lceil \frac{2}{p}\log \frac{2|\mathcal{J}|}{\epsilon} \rceil$

Enumerate all connected subsets of $\mathcal{J}$ of size at most $m$ as $polymers =[\gamma_1,\gamma_2\ldots]$

\For{$\gamma$ in $polymers$}{
    Set $\widetilde{w}(\gamma) \gets 0$ and $d^* \gets \lceil \frac{D}{2p} \log 4d \rceil $

    \For{$\ell = \lceil[(d-2d^*)|\gamma| /3^D\rceil$ \KwTo $L$}{
    Enumerate every legal Pauli path of total weight $\ell$ from immediately before $\Phi_{\gamma,d^*+1:d-d^*}^\dagger$ till immediately after, in the form $[P_{d^*},P_{d^*+1}\ldots P_{d-d^*}]$, where each Pauli operator acts on $n_\gamma$ qubits.

    \For{each $[P_{d^*},P_{d^*+1}\ldots P_{d-d^*}]$}{

        Compute $\alpha_{out} = 2^{-n_{\gamma}}\Tr[M_{\gamma}\Phi_{\gamma,1:d^*}(P_{d^*})]$ and $\alpha_{in} = \Tr[\Phi_{\gamma,d-d^*+1:d}^\dagger(P_{d-d^*})\rho_0]$ using tensor networks.
    
        Compute $\alpha_{mid} = \prod_{t = d^*+1}^{d-d^*}2^{-n_{\gamma}}\Tr(P_{t} \Phi_{\gamma,t:t}^\dagger(P_{t-1}))$ using Pauli propagation.
    
        Set $\widetilde{w}(\gamma) \gets \widetilde{w}(\gamma) + \alpha_{out}\alpha_{mid} \alpha_{in}$
        }
    }
}

Enumerate all connected multisets of elements of $polymers$, where the sum of the sizes of each element is at most $m$, as $clusters = [C_1,C_2\ldots]$\

Set $\eta \gets -|S|\log 2$

\For{$C$ in $clusters$}{
Compute the Ursell function $\varphi$ for $C$

Compute the cluster contribution $\Psi(C)$ using \cref{eq:cluster-contribution} and the previously computed values of $\varphi$ and $\widetilde{w}(\gamma) \forall \gamma \in C$.

Set $\eta \gets \eta + \Psi(C)$
}
\Return{$\eta$}\;
\end{algorithm}

\begin{lemma}[Efficiency] \label{lemma:runtime}
    Algorithm \ref{alg:fptas} runs in time $(n/\epsilon)^{O(\frac{D^{D}}{p^{D+1}})}$.
\end{lemma}
\begin{proof}
    We list the runtime costs of each step. 
    
    Step 2 costs $O(|\mathcal{J}|(e3)^{Dm}) = (n/\epsilon)^{O(D)}$ due to a standard tree-counting argument \cite{beveridge_frieze_mcdiarmid}. 
    
    Step 6 can be accomplished by first assigning the Pauli operator of minimum weight during the path. If the minimum weight is $k $, then there are $d$ locations it could be and ${n_\gamma \choose k}3^k \leq \left(3e\frac{n_{\gamma}}{k}\right)^k$ ways to assign it. Since every Pauli path has weight $\geq d-2d^*$, the Pauli path weight limit $L$ is trivial unless $d-2d^*\leq L$. Next, $\frac{n_{\gamma}}{k} \leq \frac{|\gamma|(4d)^D}{|\gamma|/3^D} = (12d)^D$. Substituting this in and using the fact that $k \leq L/(d-2d^*)$, there are $\exp(O(DL\log d/d +L))$ possible assignments of the minimum weight Pauli. Once we assign it, enumerating its path in both the forward and backward direction takes cost $\exp(L)$ (since branching can only happen at non-identity locations). Thus the overall cost is $(n/\epsilon)^{O(\frac{D}{p})}$.
    
    Step 8 involves classically simulating a circuit on at most $O(L(4d^*)^D)$ qubits since $P_{d^*}$ and $P_{d-d^*}$ have weight $O(L)$ and we only need to simulate the lightcone of depth $d^*$. With brute force simulation, this is still quasipolynomial when $d$ grows with $n$ since $d^* = O(\frac{D}{p}\log d)$ and $L= O(\frac{1}{p}\log n/\epsilon)$. We use tensor networks to obtain the improvement. We write the circuit as a tensor network with $O(L (4d^*)^{D+1})$ tensors in $D+1$ dimensions (including time). This tensor network can be contracted in time $\exp(O(L (4d^*)^{D+1})^{\frac{D}{(D+1)}}) = \exp(O(L^{\frac{D}{(D+1)}}(4d^*)^D))$ by \cite[Theorem~1]{wahl_strelchuk_tensor_network_contraction}\footnote{This is technically a randomized algorithm that applies to more general graphs than lattices (which would yield an FPRAS instead of an FPTAS), but it is possible to turn it into a deterministic algorithm for lattices using \cite{eppstein_miller_teng}}. 
    Since $d -2d^* \leq L$, we have 
    $d^* = O(\frac{D}{p}\log d)\leq O(\frac{D}{p}\log L)$. Substituting this in, we get a runtime of $\exp(O((\frac{D}{p})^DL)) = (n/\epsilon)^{O(\frac{D^{D}}{p^{D+1}})}$.

    Step 11 has runtime proportional to the number of clusters, which we bound as follows. A cluster $C$ of size $\ell$ can be encoded by a rooted plane tree
on its $\ell$ block occurrences: take a spanning tree inside each polymer and
a spanning tree of the incompatibility graph between the polymers.  The root
block has at most $|\mathcal J|$ choices.  Each tree edge has at most $3^D$
choices for the neighboring (or identical) block, and one binary label records
whether the edge lies within a polymer or connects two incompatible polymers.
Since there are at most $4^{\ell-1}$ rooted plane trees on $\ell$ vertices,
\begin{equation}\label{eq:cluster-count-bound}
\begin{aligned}
\bigl|\{C:C\text{ is a cluster and }|C|\leq m\}\bigr|
&\leq |\mathcal J|\sum_{\ell=1}^{m}4^{\ell-1}(3^D)^{\ell-1}2^{\ell-1} \\
&\leq |\mathcal J|(8\cdot 3^D)^m.
\end{aligned}
\end{equation}
Thus, the runtime of this step is $(n/\epsilon)^{O(D)}$

Step 14 costs $2^{O(m)} = O(n/\epsilon)^{O(1)}$ due to \cite[Lemma~2.6]{helmuth_perkins_regts}, as mentioned in the main text.

Thus, the dominant cost is the tensor network simulation for each Pauli path, resulting in runtime of $(n/\epsilon)^{O(\frac{1}{p}\cdot (\frac{D}{p})^D)}$

\end{proof}

\begin{lemma}[Correctness] \label{lemma:correctness}
    The output of Algorithm \ref{alg:fptas} is correct when $d > d_{crit} = O(\frac{D}{p}\log \frac{D}{p}+ \frac{3^D}{p})$.
\end{lemma}
\begin{proof}
    First, if brute force computation was used to evaluate $w(\gamma)$ (rather than steps 5-10 computing the estimate $\widetilde{w}(\gamma)$), then the estimate would achieve an accuracy of $\epsilon/2$ via \cref{corollary:accuracy}. 

    Next, we show that $\widetilde{w}(\gamma)$ achieves an additional $\epsilon/2$ error. The main intuition is that we should expect $|w(\gamma) - \widetilde{w}(\gamma)|$ to contract exponentially with $L$ during the middle layers where 2-norm contraction occurs, since it contains only Pauli paths of weight $\ell > L$ (rather than the loose lower bound of weight $\ell \geq \frac{|\gamma|(d-2d^*)}{3^D}$ used to bound $w(\gamma)$). However, in order to apply the hypercontraction and contraction chain of \cref{eq:weight-decay-hypercontractive-steps}, one needs to argue that ignoring all Pauli paths of weight $\ell \leq L$ will not increase the $\ell_2$-norm, which is complicated by the fact that Pauli paths can be weighted by signed coefficients causing cancellations. We use complex analysis to solve this issue cleanly\footnote{We expect that there are other methods, e.g. grouping Pauli paths by their accumulated weight over time and recombining them at the end, but our proof is likely simpler.}. 
    
For $\theta \in [0,2\pi]$, let $F(\theta)$ denote the value of $w(\gamma)$ when the action of each noise
channel $\mc N^{p}$ in $\Phi_{\gamma,d^*+1:d-d^*}^{\dagger}$ is replaced with a $\mc N^{1-e^{i\theta}}$ linear map (which acts on a nonidentity Pauli operator by multiplication by $e^{i\theta}$). We can express this quantity as a polynomial in $e^{i\theta}$
\[
F(\theta)=\sum_{\ell}a_\ell (e^{i\theta})^\ell
\]
Each $a_\ell$ corresponds to the signed total contribution of all weight-$\ell$ Pauli paths to $w(\gamma)$, and the noise-induced decay coefficient of $(1-p)^{\ell}$ has been replaced by $(e^{i\theta})^\ell$. Note, by Fourier orthogonality, $a_\ell = \frac{1}{2\pi}\int_0^{2\pi}
F(\theta)e^{-i\theta \ell}\,d\theta$. Next, the hypercontraction and contraction chain of \cref{eq:weight-decay-hypercontractive-steps} can be easily modified to bound $F(\theta)$ for all $\theta$: the hypercontractive steps carry over exactly while the phases preserve $\ell_2$-norm. This gives,
\begin{equation} \label{eq:r=1-bound}
    |a_\ell| \leq \frac{1}{2\pi}\int_0^{2\pi}
|F(\theta)|\,d\theta\leq \sup_\theta|F(\theta)|\leq 2^{3|\gamma|}
\end{equation}
Now, we can bound
\begin{align}
    |\tilde{w}(\gamma)| &\leq \left|\sum_{\ell > \frac{|\gamma|(d-2d^*)}{3^D}}^{L}  (1-p)^\ell a_\ell\right| \leq \frac{(1-p)^{\frac{|\gamma|(d-2d^*)}{3^D}}}{p}2^{3|\gamma|}\\
    |w(\gamma) - \widetilde{w}(\gamma)| &\leq \left|\sum_{\ell > L}^{n_{\gamma}(d-2d^*)}  (1-p)^\ell a_\ell\right| \leq \frac{(1-p)^L}{p}2^{3|\gamma|} \\
\end{align}

Using the minimum possible path weight, whether or not any
paths remain after truncation, we obtain
\[
|w(\gamma)-\widetilde w(\gamma)|
\leq \frac{2^{3|\gamma|}}{p}
(1-p)^{\max\{L,|\gamma|(d-2d^*)/3^D\}}
\leq (1-p)^{L/2}\frac{2^{3|\gamma|}}{p}
(1-p)^{|\gamma|(d-2d^*)/(2\cdot3^D)}.
\]
Reapplying the algebra at the end of \cref{sec:hypercontractivity}, we have 
 $|\widetilde{w}(\gamma)|\leq e^{-c|\gamma|}$ and $|w-\widetilde{w}(\gamma)| \leq  e^{-c|\gamma|}(1-p)^{L/2} $ when $d > d_{crit} = O(\frac{D}{p}\log \frac{D}{p}+ \frac{3^D}{p}c)$. The first inequality implies that the cluster expansion of $\widetilde{\Xi}$ (the polymer partition function
with weights $\widetilde{w}(\gamma)$) converges absolutely. We use the second inequality to replace the clusters of $\Xi$ and $\widetilde{\Xi}$ one at a time using a telescoping bound, 
\[
    \left|
        \prod_{i=1}^k w(\gamma_i)
        -\prod_{i=1}^k \widetilde{w}(\gamma_i)
    \right|
    \le
    k (1-p)^{L/2}\prod_{i=1}^k e^{-c|\gamma_i|}.
\]
Subtracting the two cluster expansions therefore gives
\begin{align}
    |\log\Xi-\log\widetilde{\Xi}|
    &\le (1-p)^{L/2}
       \sum_{k\ge1}\sum_{\gamma_1,\ldots,\gamma_k}
       k|\varphi(\gamma_1,\ldots,\gamma_k)|
       \prod_{i=1}^k e^{-c|\gamma_i|}
       \notag\\
    &\le (1-p)^{L/2}
       \sum_\gamma e^{-c|\gamma|}e^{|\gamma|}
       \notag\\
    &\le (1-p)^{L/2}|\mathcal J|
    \label{eq:log-partition-truncation-error}
\end{align}
The second inequality is
\cref{eq:ursell-summability-bound} applied to $e^{-c|\gamma|}$.
For the last inequality, every polymer contains a block,
and the convergence condition applied to the singleton
polymer $\{x\}$ gives
\[
    \sum_\gamma e^{-c|\gamma|}e^{|\gamma|}
    \le
    \sum_{x\in\mathcal J}
    \sum_{\gamma\ni x}e^{-c|\gamma|}e^{|\gamma|}
    \le|\mathcal J|.
\]
Thus taking
$L=\left\lceil\frac{2}{p}
\log\frac{2|\mathcal J|}{\epsilon}\right\rceil$
suffices to approximate $\log\Xi$ by $\log\widetilde{\Xi}$
to additive error at most $\epsilon/2$ (via a triangle inequality).
\end{proof}

\begin{proof}[Proof of \cref{thm:formal}]
    The full algorithm is specified in Algorithm \ref{alg:fptas}, the runtime is bounded in Lemma \ref{lemma:runtime}, and the accuracy is bounded in Lemma \ref{lemma:correctness}.
\end{proof}

\section{Omitted proofs}
\label{app:omitted}
\begin{proof}[Proof of Proposition \ref{prop:sampling-from-log-marginals}]
Order the qubits as $1,\ldots,n$ and generate the output bits
sequentially. Suppose that $x_S$, where $S=\{1,\ldots,i-1\}$, has
already been sampled. For each $b\in\{0,1\}$, compute an estimate
$\eta_b$ of
\[
\log P_\Phi(x_{S\cup\{i\}}),
\qquad x_i=b,
\]
such that
\[
\bigl|\eta_b-\log P_\Phi(x_{S\cup\{i\}})\bigr|
\leq\epsilon.
\]
Define the approximate conditional distribution by
\[
\widetilde{P}_\Phi(x_i=b\mid x_S)
=
\frac{e^{\eta_b}}
     {e^{\eta_0}+e^{\eta_1}}.
\]
Because exponentiating an additive $\epsilon$-approximation produces a
multiplicative $e^{\pm\epsilon}$-approximation, we have
\[
e^{-\epsilon}P_\Phi(x_{S\cup\{i\}})
\leq e^{\eta_b}
\leq e^\epsilon P_\Phi(x_{S\cup\{i\}}).
\]
The same bounds hold after summing over $b\in\{0,1\}$:
\[
e^{-\epsilon}P_\Phi(x_S)
\leq e^{\eta_0}+e^{\eta_1}
\leq e^\epsilon P_\Phi(x_S).
\]
Taking the ratio of these bounds gives
\[
e^{-2\epsilon}P_\Phi(x_i=b\mid x_S)
\leq \widetilde{P}_\Phi(x_i=b\mid x_S)
\leq e^{2\epsilon}P_\Phi(x_i=b\mid x_S).
\]
Finally, by the chain rule,
\[
P_\Phi(x)
=
\prod_{i=1}^n
P_\Phi\bigl(x_i\mid x_{\{1,\ldots,i-1\}}\bigr)
\]
and
\[
\widetilde{P}_\Phi(x)
=
\prod_{i=1}^n
\widetilde{P}_\Phi\bigl(x_i\mid x_{\{1,\ldots,i-1\}}\bigr).
\]
Multiplying the conditional bounds over all $n$ steps therefore yields
\[
e^{-2n\epsilon}P_\Phi(x)
\leq \widetilde{P}_\Phi(x)
\leq e^{2n\epsilon}P_\Phi(x).
\]
Each step requires only two marginal computations, so the total running
time is $O(nT)$. In particular, the multiplicative bound also implies
\[
\sum_{x\in\{0,1\}^n}\frac{1}{2}\left|P_\Phi(x)-\widetilde{P}_\Phi(x)\right|
\leq\frac{1}{2}
\bigl(e^{2n\epsilon}-1\bigr)
\sum_{x\in\{0,1\}^n}P_\Phi(x)
=
\frac{1}{2}(e^{2n\epsilon}-1) 
\]
\end{proof}

\begin{proof}[Proof of Fact \ref{fact:hypercontractivity}]
Quantum hypercontractivity \cite[Theorem~1]{king_hypercontractivity} is typically described using the \textit{normalized} Schatten norm $\|X\|_{q,\tau}:=2^{-n/q}\|X\|_q$. It states that 
\[
    \|\mathcal N_{[n]}^{p}(X)\|_{1+\delta',\tau}
    \leq \|X\|_{1+\delta,\tau}
    \qquad\text{whenever}\qquad
    (1-p)^2\delta'\leq\delta.
\]
Converting the above inequality from normalized to unnormalized Schatten norm, then using the fact that unital CPTP maps are contractive in every Schatten norm, and finally iterating over $d$ layers, proves the hypercontractive bound.
\end{proof}

\begin{proof}[Proof of Fact \ref{fact:contractivity}]
Let $f_0=f$ and $ f_i=(\mathcal U_i\circ\mathcal N_{[n]})(f_{i-1})$.
Write the Pauli decomposition immediately before the $i$th noise
layer as
\[
    f_{i-1}=\sum_P \alpha_P P.
\]
By the definition of $w$, every $P$ with $\alpha_P\neq0$ satisfies
$|P|\geq w$. Since
$\mathcal N_{[n]}(P)=(1-p)^{|P|}P$, orthogonality of the Pauli basis
and contractivity of unital maps in the Hilbert--Schmidt norm give
\[
\begin{aligned}
    \|f_i\|_2^2
    &\leq \|\mathcal N_{[n]}(f_{i-1})\|_2^2 \\
    &=\sum_P(1-p)^{2|P|}|\alpha_P|^2\|P\|_2^2 \\
    &\le(1-p)^{2w}\|f_{i-1}\|_2^2.
\end{aligned}
\]
Iterating over the $d$ layers proves the contractive bound.
\end{proof}

\end{document}

%% file: figures/proof-strategy-flowchart.tikz
\begingroup
\definecolor{flowink}{HTML}{17324D}
\definecolor{flowmuted}{HTML}{536B82}
\definecolor{flowaccent}{HTML}{2478B5}
\definecolor{flowaccentlight}{HTML}{EAF5FC}
\definecolor{flowcardfill}{HTML}{F8FBFD}
\definecolor{flowline}{HTML}{AFC7D8}

\tikzset{
  flow title/.style={
    font=\sffamily\bfseries\fontsize{16}{19}\selectfont,
    text=flowink,
    align=center,
    text width=8.9cm
  },
  flow subtitle/.style={
    font=\sffamily\fontsize{12}{14.5}\selectfont,
    text=flowmuted,
    align=center,
    text width=8.9cm
  },
  flow math band/.style={
    rounded corners=2.5pt,
    fill=flowaccentlight,
    draw=flowline,
    line width=0.55pt,
    minimum width=9.2cm,
    minimum height=1.25cm,
    inner xsep=4mm,
    inner ysep=2.6mm,
    font=\fontsize{15}{18}\selectfont,
    text=flowink,
    align=center
  },
  flow card/.style={
    rounded corners=6pt,
    fill=flowcardfill,
    draw=flowline,
    line width=0.8pt,
    inner xsep=4.5mm,
    inner ysep=4mm
  },
  flow branch title/.style={
    font=\sffamily\bfseries\fontsize{14}{17}\selectfont,
    text=flowink,
    align=center,
    text width=6.2cm,
    minimum height=0.75cm
  },
  flow branch subtitle/.style={
    font=\sffamily\fontsize{11.5}{14}\selectfont,
    text=flowmuted,
    align=center,
    text width=6.2cm,
    minimum height=1.05cm
  },
  flow branch band/.style={
    rounded corners=2.5pt,
    fill=flowaccentlight,
    draw=flowline,
    line width=0.55pt,
    minimum width=6.6cm,
    text width=6.2cm,
    minimum height=1.05cm,
    inner xsep=2mm,
    inner ysep=2.2mm,
    font=\sffamily\bfseries\fontsize{10.5}{13}\selectfont,
    text=flowaccent,
    align=center
  },
  flow arrow/.style={
    -{Stealth[length=3.2mm,width=2.4mm]},
    draw=flowaccent,
    line width=1.15pt,
    rounded corners=3pt
  },
  flow line/.style={
    draw=flowaccent,
    line width=1.15pt,
    rounded corners=3pt
  }
}

\begin{tikzpicture}
  \node[flow title] (sample-title) {Approximate sampling};
  \node[flow subtitle, below=0.5mm of sample-title]
    (sample-subtitle) {Draw from a distribution close to the circuit output};
  \node[flow math band, below=2.2mm of sample-subtitle]
    (sample-math) {$\widetilde P(x)=(1\pm\varepsilon)P(x)$};
  \begin{scope}[on background layer]
    \node[flow card, fit=(sample-title)(sample-subtitle)(sample-math)] (sample) {};
  \end{scope}

  \node[flow title, below=15mm of sample]
    (marginal-title) {Sampling-to-computing reduction};
  \node[flow subtitle, below=0.5mm of marginal-title]
    (marginal-subtitle) {Compute the log-marginals};
  \node[flow math band, below=2.2mm of marginal-subtitle]
    (marginal-math) {$\log P(x_S)\;\pm\;\varepsilon$};
  \begin{scope}[on background layer]
    \node[flow card, fit=(marginal-title)(marginal-subtitle)(marginal-math)] (marginal) {};
  \end{scope}

  \node[flow title, below=15mm of marginal]
    (cluster-title) {Truncated cluster expansion};
  \node[flow subtitle, below=0.5mm of cluster-title]
    (cluster-subtitle) {Express the log-marginal using only small clusters};
  \node[flow math band, below=2.2mm of cluster-subtitle]
    (cluster-math) {$\displaystyle
      \log P(x_S)\;\approx\!
      \sum_{\substack{C:\,|C|\le O(\!\log n\!)}}\!\Psi(C)
      \;\pm\;\varepsilon$};
  \begin{scope}[on background layer]
    \node[flow card, fit=(cluster-title)(cluster-subtitle)(cluster-math)] (cluster) {};
  \end{scope}

  \coordinate (fork) at ($(cluster.south)+(0,-12mm)$);

  \node[flow branch title, anchor=north] at ($(fork)+(-4.25cm,-16mm)$)
    (kp-title) {Prove convergence};
  \node[flow branch subtitle, below=0.6mm of kp-title]
    (kp-subtitle) {Verify the Koteck\'y--Preiss condition};
  \node[flow branch band, below=2.2mm of kp-subtitle]
    (kp-band) {via hypercontractivity};
  \begin{scope}[on background layer]
    \node[flow card, fit=(kp-title)(kp-subtitle)(kp-band)] (kp) {};
  \end{scope}

  \node[flow branch title, anchor=north] at ($(fork)+(4.25cm,-16mm)$)
    (pauli-title) {Compute clusters $\Psi(C)$};
  \node[flow branch subtitle, below=0.6mm of pauli-title]
    (pauli-subtitle) {using brute-force or};
  \node[flow branch band, below=2.2mm of pauli-subtitle]
    (pauli-band) {via tensor network $+$ Pauli paths};
  \begin{scope}[on background layer]
    \node[flow card, fit=(pauli-title)(pauli-subtitle)(pauli-band)] (pauli) {};
  \end{scope}

  \draw[flow arrow] (sample.south) -- (marginal.north);
  \draw[flow arrow] (marginal.south) -- (cluster.north);
  \draw[flow line] (cluster.south) -- (fork);
  \draw[flow arrow] (fork) -| (kp.north);
  \draw[flow arrow] (fork) -| (pauli.north);
\end{tikzpicture}
\endgroup

%% file: figures/backward-lightcone.tikz
\begin{tikzpicture}[
    x={(1cm,0cm)},
    y={(0.40cm,0.34cm)},
    z={(0cm,0.95cm)},
    line cap=round,
    line join=round
]
    \definecolor{spacetimegreen}{HTML}{8FD5B5}
    \definecolor{coneblue}{HTML}{2878B5}
    \definecolor{coneorange}{HTML}{D9822B}
    \colorlet{coneblend}{coneblue!50!coneorange}
    \definecolor{blockred}{HTML}{D62828}

    \path[fill=gray!14,draw=gray!60,line width=0.7pt]
        (0,0,0) -- (7,0,0) -- (7,0,1.24) -- (0,0,1.24) -- cycle;
    \path[fill=gray!24,draw=gray!60,line width=0.7pt]
        (7,0,0) -- (7,5,0) -- (7,5,1.24) -- (7,0,1.24) -- cycle;

    \path[fill=gray!5,draw=gray!55,line width=0.7pt]
        (0,0,0) -- (7,0,0) -- (7,5,0) -- (0,5,0) -- cycle;
    \foreach \xx in {1,...,6}
        \draw[gray!48,opacity=0.23,line width=0.42pt]
            (\xx,0,0) -- (\xx,5,0);
    \foreach \yy in {1,...,4}
        \draw[gray!48,opacity=0.23,line width=0.42pt]
            (0,\yy,0) -- (7,\yy,0);
    \draw[gray!60,dashed,line width=0.65pt]
        (0,5,0) -- (7,5,0);
    \draw[gray!60,dashed,line width=0.65pt]
        (0,5,0) -- (0,5,1.24);
    \path[fill=coneblue!18,draw=none]
        (1.5,1.5,0) -- (3.44,1.5,0) -- (3.44,3.5,0) -- (1.5,3.5,0) -- cycle;
    \path[fill=coneblue!28,draw=none]
        (1.5,1.5,0) -- (3.44,1.5,0) -- (3,2,1.24) -- (2,2,1.24) -- cycle;
    \path[fill=coneblue!34,draw=none]
        (3.44,1.5,0) -- (3.44,3.5,0) -- (3,3,1.24) -- (3,2,1.24) -- cycle;
    \path[fill=coneblue!22,draw=none]
        (3.44,3.5,0) -- (1.5,3.5,0) -- (2,3,1.24) -- (3,3,1.24) -- cycle;
    \path[fill=coneblue!16,draw=none]
        (1.5,3.5,0) -- (1.5,1.5,0) -- (2,2,1.24) -- (2,3,1.24) -- cycle;

    \path[fill=coneorange!18,draw=none]
        (3.56,1.5,0) -- (5.5,1.5,0) -- (5.5,3.5,0) -- (3.56,3.5,0) -- cycle;
    \path[fill=coneorange!28,draw=none]
        (3.56,1.5,0) -- (5.5,1.5,0) -- (5,2,1.24) -- (4,2,1.24) -- cycle;
    \path[fill=coneorange!34,draw=none]
        (5.5,1.5,0) -- (5.5,3.5,0) -- (5,3,1.24) -- (5,2,1.24) -- cycle;
    \path[fill=coneorange!22,draw=none]
        (5.5,3.5,0) -- (3.56,3.5,0) -- (4,3,1.24) -- (5,3,1.24) -- cycle;
    \path[fill=coneorange!16,draw=none]
        (3.56,3.5,0) -- (3.56,1.5,0) -- (4,2,1.24) -- (4,3,1.24) -- cycle;

    \path[fill=coneblend!34,draw=none]
        (3.44,2.003,0) -- (3.44,3.5,0) --
        (3.2995,3.3404,0.3959) -- cycle;

    \path[fill=spacetimegreen!82,fill opacity=0.62,
          draw=spacetimegreen!45!black,draw opacity=1,line width=0.8pt]
        (0,0,1.24) -- (7,0,1.24) -- (7,5,1.24) -- (0,5,1.24) -- cycle;

    \foreach \xx in {1,...,6}
        \draw[spacetimegreen!35!black,opacity=0.58,line width=0.82pt]
            (\xx,0,1.24) -- (\xx,5,1.24);
    \foreach \yy in {1,...,4}
        \draw[spacetimegreen!35!black,opacity=0.58,line width=0.82pt]
            (0,\yy,1.24) -- (7,\yy,1.24);

    \draw[draw=coneblue!80!black,line width=0.7pt]
        (1.5,3.5,0) -- (1.5,1.5,0) -- (3.44,1.5,0) -- (3.44,3.5,0);
    \draw[draw=coneblue!80!black,line width=0.7pt]
        (1.5,3.5,0) -- (3.44,3.5,0);
    \draw[draw=coneblue!70!black,dashed,line width=0.65pt]
        (2,2,1.24) -- (1.5,1.5,0)
        (3,2,1.24) -- (3.44,1.5,0);
    \draw[draw=coneblue!65!black,dashed,line width=0.65pt]
        (2,3,1.24) -- (1.5,3.5,0)
        (3,3,1.24) -- (3.44,3.5,0);

    \draw[draw=coneorange!75!black,line width=0.7pt]
        (3.56,3.5,0) -- (3.56,1.5,0) -- (5.5,1.5,0) -- (5.5,3.5,0);
    \draw[draw=coneorange!75!black,line width=0.7pt]
        (3.56,3.5,0) -- (5.5,3.5,0);
    \draw[draw=coneorange!68!black,dashed,line width=0.65pt]
        (4,2,1.24) -- (3.56,1.5,0)
        (5,2,1.24) -- (5.5,1.5,0);
    \draw[draw=coneorange!65!black,dashed,line width=0.65pt]
        (4,3,1.24) -- (3.56,3.5,0)
        (5,3,1.24) -- (5.5,3.5,0);

    \path[fill=blockred,draw=blockred!55!black,line width=0.9pt]
        (2,2,1.24) -- (3,2,1.24) -- (3,3,1.24) -- (2,3,1.24) -- cycle;
    \path[fill=blockred,draw=blockred!55!black,line width=0.9pt]
        (4,2,1.24) -- (5,2,1.24) -- (5,3,1.24) -- (4,3,1.24) -- cycle;

    \foreach \bx in {0,...,6}
        \foreach \by in {0,...,4}
            \foreach \dx in {0.22,0.50,0.78}
                \foreach \dy in {0.22,0.50,0.78}
                    \filldraw[fill=black!72,draw=white,line width=0.12pt]
                        ({\bx+\dx},{\by+\dy},1.24) circle[radius=0.026cm];

    \draw[gray!60,line width=0.75pt]
        (7,0,0) -- (7,0,1.24);

\end{tikzpicture}

%% file: figures/lattice-animals.tikz
\begin{tikzpicture}[
    x={(0.62cm,0cm)},
    y={(0.248cm,0.211cm)},
    z={(0cm,0.59cm)},
    line cap=round,
    line join=round
]
    \definecolor{spacetimegreen}{HTML}{8FD5B5}
    \definecolor{animalblue}{HTML}{3478B8}
    \definecolor{animalorange}{HTML}{E28A2B}
    \definecolor{animalpurple}{HTML}{8A63B8}
    \definecolor{animalred}{HTML}{C94F5C}

    \path[fill=gray!13,draw=gray!55,line width=0.65pt]
        (0,0,0) -- (9,0,0) -- (9,0,1) -- (0,0,1) -- cycle;
    \path[fill=gray!22,draw=gray!55,line width=0.65pt]
        (9,0,0) -- (9,8,0) -- (9,8,1) -- (9,0,1) -- cycle;
    \path[fill=gray!4,draw=gray!50,line width=0.65pt]
        (0,0,0) -- (9,0,0) -- (9,8,0) -- (0,8,0) -- cycle;

    \def\animalcone#1#2#3#4#5{%
        \path[fill=#1!17,draw=#1!58!black,line width=0.42pt]
            ({#2-0.18},{#3-0.18},0) -- ({#4+0.18},{#3-0.18},0) --
            ({#4+0.18},{#5+0.18},0) -- ({#2-0.18},{#5+0.18},0) -- cycle;
        \path[fill=#1!27,draw=#1!62!black,line width=0.42pt]
            ({#2-0.18},{#3-0.18},0) -- ({#4+0.18},{#3-0.18},0) --
            (#4,#3,1) -- (#2,#3,1) -- cycle;
        \path[fill=#1!31,draw=#1!62!black,line width=0.42pt]
            ({#4+0.18},{#3-0.18},0) -- ({#4+0.18},{#5+0.18},0) --
            (#4,#5,1) -- (#4,#3,1) -- cycle;
        \path[fill=#1!20,draw=#1!58!black,line width=0.42pt]
            ({#4+0.18},{#5+0.18},0) -- ({#2-0.18},{#5+0.18},0) --
            (#2,#5,1) -- (#4,#5,1) -- cycle;
        \path[fill=#1!15,draw=#1!58!black,line width=0.42pt]
            ({#2-0.18},{#5+0.18},0) -- ({#2-0.18},{#3-0.18},0) --
            (#2,#3,1) -- (#2,#5,1) -- cycle;
    }

    \animalcone{animalblue}{1}{1}{3}{4}
    \animalcone{animalorange}{5}{1}{8}{4}
    \animalcone{animalred}{1}{5}{3}{7}
    \animalcone{animalpurple}{4}{5}{7}{7}

    \path[fill=spacetimegreen!82,fill opacity=0.62,
          draw=spacetimegreen!42!black,draw opacity=1,line width=0.8pt]
        (0,0,1) -- (9,0,1) -- (9,8,1) -- (0,8,1) -- cycle;
    \foreach \x in {1,...,8}
        \draw[spacetimegreen!35!black,opacity=0.58,line width=0.82pt]
            (\x,0,1) -- (\x,8,1);
    \foreach \y in {1,...,7}
        \draw[spacetimegreen!35!black,opacity=0.58,line width=0.82pt]
            (0,\y,1) -- (9,\y,1);

    \foreach \x/\y in {1/1,2/1,1/2,1/3,2/3}
        \path[fill=animalblue!82,draw=animalblue!55!black,line width=0.45pt]
            (\x,\y,1) -- ++(1,0,0) -- ++(0,1,0) -- ++(-1,0,0) -- cycle;
    \foreach \x/\y in {5/1,5/2,6/2,7/2,7/3}
        \path[fill=animalorange!84,draw=animalorange!58!black,line width=0.45pt]
            (\x,\y,1) -- ++(1,0,0) -- ++(0,1,0) -- ++(-1,0,0) -- cycle;
    \foreach \x/\y in {4/6,5/6,6/6,5/5,6/5}
        \path[fill=animalpurple!82,draw=animalpurple!55!black,line width=0.45pt]
            (\x,\y,1) -- ++(1,0,0) -- ++(0,1,0) -- ++(-1,0,0) -- cycle;
    \foreach \x/\y in {1/6,2/6,1/5}
        \path[fill=animalred!82,draw=animalred!55!black,line width=0.45pt]
            (\x,\y,1) -- ++(1,0,0) -- ++(0,1,0) -- ++(-1,0,0) -- cycle;

    \foreach \bx in {0,...,8}
        \foreach \by in {0,...,7}
            \foreach \dx in {0.22,0.50,0.78}
                \foreach \dy in {0.22,0.50,0.78}
                    \filldraw[fill=black!72,draw=white,line width=0.10pt]
                        ({\bx+\dx},{\by+\dy},1) circle[radius=0.024cm];

    \draw[gray!58,line width=0.65pt] (9,0,0) -- (9,0,1);
\end{tikzpicture}

%% file: figures/contractivity_layers.tikz
\begin{tikzpicture}[x=0.88cm,y=0.88cm,>=Latex,
 gate/.style={draw=wire!75,fill=white,rounded corners=1pt,line width=.5pt,minimum height=3mm},
 norm/.style={fill=white,inner sep=3pt,font=\large},
 brace/.style={decorate,decoration={brace,amplitude=5pt},line width=.7pt},
 every node/.style={font=\small}]
\fill[hc!9,rounded corners=3pt] (-.35,4.15) rectangle (7.35,2.55);
\fill[ct!9,rounded corners=3pt] (-.35,10.25) rectangle (7.35,4.15);
\fill[hc!9,rounded corners=3pt] (-.35,11.85) rectangle (7.35,10.25);
\foreach \x in {0,...,7} {\draw[wire!65,line width=.55pt] (\x,11.85)--(\x,2.55);}
\foreach \y in {3,4.75,6.3,7.85,9.4,10.75} {
 \foreach \x in {0,2,4,6} {\draw[gate] (\x-.13,\y-.14) rectangle (\x+1.13,\y+.14);}}
\foreach \y in {3.7,5.5,7.05,8.6,9.95,11.45} {
 \foreach \x in {1,3,5} {\draw[gate] (\x-.13,\y-.14) rectangle (\x+1.13,\y+.14);}}
\foreach \y in {3.35,3.95,5.1,5.9,6.65,7.45,8.2,9,9.65,11.1,11.7} {
 \foreach \x in {0,...,7} {\fill[wire!65] (\x,\y) circle (1.1pt);}}
\foreach \y in {2.55,4.15,10.25,11.85} {\draw[wire!35,dashed] (-.4,\y)--(8.3,\y);}
\draw[->,wire,thick] (-2.1,2.55)--(-2.1,11.85)
 node[midway,left=5pt,rotate=90,anchor=south] {time};
\draw[brace] (-.7,2.55)--(-.7,4.15) node[midway,left=8pt] {$d^*$};
\draw[brace] (-.7,4.15)--(-.7,10.25) node[midway,left=8pt,rotate=90,anchor=south] {$d-2d^*$};
\draw[brace] (-.7,10.25)--(-.7,11.85) node[midway,left=8pt] {$d^*$};
\node[anchor=north] at (3.5,2.45) {$\Phi_{\gamma,1:d}^{\dagger}(M_\gamma)$};
\node[anchor=south] at (3.5,12) {$M_\gamma$};
\node[norm] (inf) at (9.15,0.95) {$\|\cdot\|_\infty$};
\node[norm] (q) at (9.15,2.55) {$\|\cdot\|_{1+(4d)^D}$};
\node[norm] (twoA) at (9.15,4.15) {$\|\cdot\|_2$};
\node[norm] (twoB) at (9.15,10.25) {$\|\cdot\|_2$};
\node[norm] (r) at (9.15,11.85) {$\|\cdot\|_{1+(4d)^{-D}}$};
\node[norm] (one) at (9.15,12.75) {$\|M_\gamma\|_1$};
\draw[->,wire] (inf.north)--(q.south) node[midway,right=11pt,font=\scriptsize] {monotonicity};
\draw[->,hc,thick] (q)--(twoA);
\draw[->,ct,thick] (twoA)--(twoB);
\draw[->,hc,thick] (twoB)--(r);
\draw[->,wire] (r)--(one) node[midway,right=20pt,font=\scriptsize] {monotonicity};
\node[anchor=west,text=hc,align=left] at (10.25,3.35) {Hypercontractivity\\[3pt]$\displaystyle 2^{-n_\gamma(\frac12-\frac{1}{1+(4d)^D})}$};
\node[anchor=west,text=ct,align=left] at (10.25,7.2) {$2$-norm contraction\\[5pt]$\displaystyle (1-p)^{|\gamma|(d-2d^*)/(3^D)}$};
\node[anchor=west,text=hc,align=left] at (10.25,11.05) {Hypercontractivity\\[3pt]$\displaystyle 2^{-n_\gamma(\frac{1}{1+(4d)^{-D}}-\frac12)}$};
\node[anchor=west] at (10.6,12.75) {$\leq 2^{n_\gamma+|\gamma|}$};
\end{tikzpicture}

%% file: joel.bib
@misc{scalet_entanglement_sudden_death,
      title={Spatial Entanglement Sudden Death in Spin Chains at All Temperatures}, 
      author={Samuel O. Scalet},
      year={2026},
      eprint={2602.20694},
      archivePrefix={arXiv},
      primaryClass={quant-ph},
      url={https://arxiv.org/abs/2602.20694}, 
}

@misc{bakshi_spin_chains,
      title={Entanglement in quantum spin chains is strictly finite at any temperature}, 
      author={Ainesh Bakshi and Soonwon Choi and Saúl Pilatowsky-Cameo},
      year={2026},
      eprint={2602.13386},
      archivePrefix={arXiv},
      primaryClass={quant-ph},
      url={https://arxiv.org/abs/2602.13386}, 
}

@misc{shravan_iqp_damping,
      title={Efficient simulation of noisy IQP circuits with amplitude-damping noise}, 
      author={Shravan Shravan and Mohsin Raza and Ariel Shlosberg},
      year={2026},
      eprint={2604.05036},
      archivePrefix={arXiv},
      primaryClass={quant-ph},
      url={https://arxiv.org/abs/2604.05036}, 
}

@misc{wei_rajakumar_nelson_noise-induced_trunctation,
      title={Noise-induced contraction of MPO truncation errors in noisy random circuits and Lindbladian dynamics}, 
      author={Zhi-Yuan Wei and Joel Rajakumar and Jon Nelson and Daniel Malz and Michael J. Gullans and Alexey V. Gorshkov},
      year={2026},
      eprint={2603.20400},
      archivePrefix={arXiv},
      primaryClass={quant-ph},
      url={https://arxiv.org/abs/2603.20400}, 
}

@misc{lee_ghosh_approx_markov,
      title={Classical simulation of noisy random circuits from exponential decay of correlation}, 
      author={Su-un Lee and Soumik Ghosh and Changhun Oh and Kyungjoo Noh and Bill Fefferman and Liang Jiang},
      year={2025},
      eprint={2510.06328},
      archivePrefix={arXiv},
      primaryClass={quant-ph},
      url={https://arxiv.org/abs/2510.06328}, 
}

@misc{zhang_lee_approx_markov,
      title={Classically Sampling Noisy Quantum Circuits in Quasi-Polynomial Time under Approximate Markovianity}, 
      author={Yifan F. Zhang and Su-un Lee and Liang Jiang and Sarang Gopalakrishnan},
      year={2025},
      eprint={2510.06324},
      archivePrefix={arXiv},
      primaryClass={quant-ph},
      url={https://arxiv.org/abs/2510.06324}, 
}

@misc{nelson_rajakumar_limitations,
      title={Limitations of Noisy Geometrically Local Quantum Circuits}, 
      author={Jon Nelson and Joel Rajakumar and Michael J. Gullans},
      year={2025},
      eprint={2510.06346},
      archivePrefix={arXiv},
      primaryClass={quant-ph},
      url={https://arxiv.org/abs/2510.06346}, 
}

@misc{anshu_liu_nguyen_pattison_loglogn_iqp,
author={A. Anshu and Y. Liu and Q. Nguyen and C. Pattison},
note = {Simons talk available at: https://simons.berkeley.edu/talks/quynh-nguyen-harvard-university-2025-05-28},
year = {2025}
}

@misc{park_oh_mps_noisy_iqp,
      title={Matrix product state approach to lossy boson sampling and noisy IQP sampling}, 
      author={Sojeong Park and Changhun Oh},
      year={2025},
      eprint={2510.24137},
      archivePrefix={arXiv},
      primaryClass={quant-ph},
      url={https://arxiv.org/abs/2510.24137}, 
}

@misc{zlokapa_syk,
      title={SYK thermal expectations are classically easy at any temperature}, 
      author={Alexander Zlokapa and Bobak T. Kiani},
      year={2026},
      eprint={2602.22619},
      archivePrefix={arXiv},
      primaryClass={quant-ph},
      url={https://arxiv.org/abs/2602.22619}, 
}

@article{wei_nelwon_rajakumar_2d_clifford,
  title = {Measurement-Induced Entanglement in Noisy 2D Random Circuits},
  author = {Wei, Zhi-Yuan and Nelson, Jon and Rajakumar, Joel and Cruz, Esther and Gorshkov, Alexey V. and Gullans, Michael J. and Malz, Daniel},
  journal = {PRX Quantum},
  volume = {7},
  number = {3},
  pages = {033023},
  year = {2026},
  month = aug,
  doi = {10.1103/yfzm-5rr6},
  url = {https://doi.org/10.1103/yfzm-5rr6}
}

@article{garcia_cirac_trivedi_pauli,
  title = {Pauli Path Simulations of Noisy Quantum Circuits beyond Average Case},
  author = {Gonz{\'a}lez-Garc{\'i}a, Guillermo and Cirac, J. Ignacio and Trivedi, Rahul},
  journal = {Quantum},
  volume = {9},
  pages = {1730},
  year = {2025},
  month = may,
  doi = {10.22331/q-2025-05-05-1730},
  url = {https://doi.org/10.22331/q-2025-05-05-1730}
}

@article{schuster_polynomial,
  title = {A Polynomial-Time Classical Algorithm for Noisy Quantum Circuits},
  author = {Schuster, Thomas and Yin, Chao and Gao, Xun and Yao, Norman Y.},
  journal = {Physical Review X},
  volume = {15},
  number = {4},
  pages = {041018},
  year = {2025},
  month = nov,
  doi = {10.1103/xct1-7kf2},
  url = {https://doi.org/10.1103/xct1-7kf2}
}

@article{angrisani_schmidhuber_rudolph_cerezo_holmes_huang,
  title = {Classically Estimating Observables of Noiseless Quantum Circuits},
  author = {Angrisani, Armando and Schmidhuber, Alexander and Rudolph, Manuel S. and Cerezo, M. and Holmes, Zo{\"e} and Huang, Hsin-Yuan},
  journal = {Physical Review Letters},
  volume = {135},
  number = {17},
  pages = {170602},
  year = {2025},
  month = oct,
  doi = {10.1103/lh6x-7rc3},
  url = {https://doi.org/10.1103/lh6x-7rc3}
}

@article{angrisani_mele_simulating,
  title = {Simulating Quantum Circuits with Arbitrary Local Noise Using Pauli Propagation},
  author = {Angrisani, Armando and Mele, Antonio A. and Rudolph, Manuel S. and Cerezo, M. and Holmes, Zo{\"e}},
  journal = {PRX Quantum},
  volume = {7},
  number = {2},
  pages = {020313},
  year = {2026},
  month = apr,
  doi = {10.1103/fb28-wlv2},
  url = {https://doi.org/10.1103/fb28-wlv2}
}

@article{mele_noise-induced,
  title = {Noise-Induced Shallow Circuits and the Absence of Barren Plateaus},
  author = {Mele, Antonio Anna and Angrisani, Armando and Ghosh, Soumik and Khatri, Sumeet and Eisert, Jens and Stilck Fran{\c{c}}a, Daniel and Quek, Yihui},
  journal = {Nature Physics},
  volume = {22},
  number = {5},
  pages = {751--756},
  year = {2026},
  doi = {10.1038/s41567-026-03245-z},
  url = {https://doi.org/10.1038/s41567-026-03245-z}
}

@article{shtanko_sharma_nonunital_1D,
  title = {Complexity of Local Quantum Circuits under Nonunital Noise},
  author = {Shtanko, Oles and Sharma, Kunal},
  journal = {PRX Quantum},
  volume = {6},
  number = {3},
  pages = {030347},
  year = {2025},
  month = sep,
  doi = {10.1103/xwz2-wfr4},
  url = {https://doi.org/10.1103/xwz2-wfr4}
}

@article{quantinuum_arbitrary_geometries,
  title = {The Computational Power of Random Quantum Circuits in Arbitrary Geometries},
  author = {DeCross, Matthew and Haghshenas, Reza and Liu, Minzhao and Rinaldi, Enrico and Gray, Johnnie and Alexeev, Yuri and Baldwin, Charles H. and Bartolotta, John P. and Bohn, Matthew and Chertkov, Eli and Cline, Julia and Colina, Jonhas and DelVento, Davide and Dreiling, Joan M. and Foltz, Cameron and Gaebler, John P. and Gatterman, Thomas M. and Gilbreth, Christopher N. and Giles, Joshua and Gresh, Dan and Hall, Alex and Hankin, Aaron and Hansen, Azure and Hewitt, Nathan and Hoffman, Ian and Holliman, Craig and Hutson, Ross B. and Jacobs, Trent and Johansen, Jacob and Lee, Patricia J. and Lehman, Elliot and Lucchetti, Dominic and Lykov, Danylo and Madjarov, Ivaylo S. and Mathewson, Brian and Mayer, Karl and Mills, Michael and Niroula, Pradeep and Pino, Juan M. and Roman, Conrad and Schecter, Michael and Siegfried, Peter E. and Tiemann, Bruce G. and Volin, Curtis and Walker, James and Shaydulin, Ruslan and Pistoia, Marco and Moses, Steven A. and Hayes, David and Neyenhuis, Brian and Stutz, Russell P. and Foss-Feig, Michael},
  journal = {Physical Review X},
  volume = {15},
  number = {2},
  pages = {021052},
  year = {2025},
  month = may,
  doi = {10.1103/PhysRevX.15.021052},
  url = {https://doi.org/10.1103/PhysRevX.15.021052}
}

@article{rajakumar_watson,
   title={Gibbs Sampling gives Quantum Advantage at Constant Temperatures with O(1)-Local Hamiltonians},
   volume={10},
   ISSN={2521-327X},
   url={http://dx.doi.org/10.22331/q-2026-01-22-1981},
   DOI={10.22331/q-2026-01-22-1981},
   journal={Quantum},
   publisher={Verein zur Forderung des Open Access Publizierens in den Quantenwissenschaften},
   author={Rajakumar, Joel and Watson, James D.},
   year={2026},
   month=Jan, pages={1981} }

@inbook{nelson_rajakumar_naturally,
author = {Jon Nelson and Joel Rajakumar and Dominik Hangleiter and Michael J. Gullans},
title = {Polynomial-Time Classical Simulation of Noisy Quantum Circuits with Naturally Fault-Tolerant Gates},
booktitle = {Proceedings of the 2026 Annual ACM-SIAM Symposium on Discrete Algorithms (SODA)},
chapter = {},
pages = {1309-1331},
doi = {10.1137/1.9781611978971.50},
URL = {https://epubs.siam.org/doi/abs/10.1137/1.9781611978971.50},
eprint = {https://epubs.siam.org/doi/pdf/10.1137/1.9781611978971.50}
}

@Article{chiu_continuous_operation_3000-qubit_system,
author={Chiu, Neng-Chun
and Trapp, Elias C.
and Guo, Jinen
and Abobeih, Mohamed H.
and Stewart, Luke M.
and Hollerith, Simon
and Stroganov, Pavel L.
and Kalinowski, Marcin
and Geim, Alexandra A.
and Evered, Simon J.
and Li, Sophie H.
and Lyu, Xingjian
and Peters, Lisa M.
and Bluvstein, Dolev
and Wang, Tout T.
and Greiner, Markus
and Vuleti{\'{c}}, Vladan
and Lukin, Mikhail D.},
title={Continuous operation of a coherent 3,000-qubit system},
journal={Nature},
year={2025},
month={Oct},
day={01},
volume={646},
number={8087},
pages={1075-1080},
issn={1476-4687},
doi={10.1038/s41586-025-09596-6},
url={https://doi.org/10.1038/s41586-025-09596-6}
}

@article{shor_factoring_1997,
  author  = {Shor, Peter W.},
  title   = {Polynomial-Time Algorithms for Prime Factorization and Discrete Logarithms on a Quantum Computer},
  journal = {SIAM Journal on Computing},
  volume  = {26},
  number  = {5},
  pages   = {1484--1509},
  year    = {1997},
  doi     = {10.1137/S0097539795293172}
}

@article{lloyd_quantum_simulators_1996,
  author  = {Lloyd, Seth},
  title   = {Universal Quantum Simulators},
  journal = {Science},
  volume  = {273},
  number  = {5278},
  pages   = {1073--1078},
  year    = {1996},
  doi     = {10.1126/science.273.5278.1073}
}

@article{knill_noisy_devices_2005,
  author  = {Knill, Emanuel},
  title   = {Quantum computing with realistically noisy devices},
  journal = {Nature},
  volume  = {434},
  number  = {7029},
  pages   = {39--44},
  year    = {2005},
  doi     = {10.1038/nature03350}
}

@article{geher_reset_2025,
  author  = {Geh{\'e}r, Gy{\"o}rgy P. and Jastrzebski, Marcin and Campbell, Earl T. and Crawford, Ophelia},
  title   = {To reset, or not to reset---that is the question},
  journal = {npj Quantum Information},
  volume  = {11},
  pages   = {39},
  year    = {2025},
  doi     = {10.1038/s41534-025-00998-y}
}

@article{decross_qubit_reuse_2023,
  author  = {DeCross, Matthew and Chertkov, Eli and Kohagen, Megan and Foss-Feig, Michael},
  title   = {Qubit-Reuse Compilation with Mid-Circuit Measurement and Reset},
  journal = {Physical Review X},
  volume  = {13},
  pages   = {041057},
  year    = {2023},
  doi     = {10.1103/PhysRevX.13.041057}
}

@article{mann_efficient_algorithms_approximating,
  title={Efficient algorithms for approximating quantum partition functions},
  author={Mann, Ryan L and Helmuth, Tyler},
  journal={Journal of Mathematical Physics},
  volume={62},
  number={2},
  year={2021},
  publisher={AIP Publishing},
  doi={https://doi.org/10.22331/q-2023-10-25-1155}
}

@article{yin_polynomial-time_classical_sampling,
  title={Polynomial-time classical sampling of high-temperature quantum Gibbs states},
  author={Yin, Chao and Lucas, Andrew},
  journal={arXiv preprint arXiv:2305.18514},
  year={2023},
  doi={https://doi.org/10.48550/arXiv.2305.18514}
}

@inproceedings{harrow_classical_algorithms_correlation,
  title={Classical algorithms, correlation decay, and complex zeros of partition functions of quantum many-body systems},
  author={Harrow, Aram W and Mehraban, Saeed and Soleimanifar, Mehdi},
  booktitle={Proceedings of the 52nd Annual ACM SIGACT Symposium on Theory of Computing},
  pages={378--386},
  year={2020},
  doi={https://doi.org/10.1145/3357713.3384322}
}

@article{nonunital_fefferman_ghosh_soumik_gullans,
  title = {Effect of Nonunital Noise on Random-Circuit Sampling},
  author = {Fefferman, Bill and Ghosh, Soumik and Gullans, Michael and Kuroiwa, Kohdai and Sharma, Kunal},
  journal = {PRX Quantum},
  volume = {5},
  issue = {3},
  pages = {030317},
  numpages = {41},
  year = {2024},
  month = {Jul},
  publisher = {American Physical Society},
  doi = {10.1103/PRXQuantum.5.030317},
  url = {https://link.aps.org/doi/10.1103/PRXQuantum.5.030317}
}

@article{raussendorf_bravyi_harrington_cluster_states,
  title = {Long-range quantum entanglement in noisy cluster states},
  author = {Raussendorf, Robert and Bravyi, Sergey and Harrington, Jim},
  journal = {Phys. Rev. A},
  volume = {71},
  issue = {6},
  pages = {062313},
  numpages = {6},
  year = {2005},
  month = {Jun},
  publisher = {American Physical Society},
  doi = {10.1103/PhysRevA.71.062313},
  url = {https://link.aps.org/doi/10.1103/PhysRevA.71.062313}
}

@article{fontana,
	Author = {Fontana, Enrico and Rudolph, Manuel S. and Duncan, Ross and Rungger, Ivan and C{\^\i}rstoiu, Cristina},
	Journal = {npj Quantum Information},
	Number = {1},
	Pages = {84},
	Title = {Classical simulations of noisy variational quantum circuits},
	Volume = {11},
	Year = {2025}}

@article{watts_gosset_liu_soleimanifar,
  title = {Quantum Advantage from Measurement-Induced Entanglement in Random Shallow Circuits},
  author = {Bene Watts, Adam and Gosset, David and Liu, Yinchen and Soleimanifar, Mehdi},
  journal = {PRX Quantum},
  volume = {6},
  issue = {1},
  pages = {010356},
  numpages = {25},
  year = {2025},
  month = {Mar},
  publisher = {American Physical Society},
  doi = {10.1103/PRXQuantum.6.010356},
  url = {https://link.aps.org/doi/10.1103/PRXQuantum.6.010356}
}

@article{hangleiter_fault-tolerant,
  title = {Fault-Tolerant Compiling of Classically Hard Instantaneous Quantum Polynomial Circuits on Hypercubes},
  author = {Hangleiter, Dominik and Kalinowski, Marcin and Bluvstein, Dolev and Cain, Madelyn and Maskara, Nishad and Gao, Xun and Kubica, Aleksander and Lukin, Mikhail D. and Gullans, Michael J.},
  journal = {PRX Quantum},
  volume = {6},
  issue = {2},
  pages = {020338},
  numpages = {50},
  year = {2025},
  month = {May},
  publisher = {American Physical Society},
  doi = {10.1103/PRXQuantum.6.020338},
  url = {https://link.aps.org/doi/10.1103/PRXQuantum.6.020338}
}

@article{bravyi_gosset_koenig_tomamichel,
   title={Quantum advantage with noisy shallow circuits},
   volume={16},
   ISSN={1745-2481},
   url={http://dx.doi.org/10.1038/s41567-020-0948-z},
   DOI={10.1038/s41567-020-0948-z},
   number={10},
   journal={Nature Physics},
   publisher={Springer Science and Business Media LLC},
   author={Bravyi, Sergey and Gosset, David and König, Robert and Tomamichel, Marco},
   year={2020},
   month=jul, pages={1040–1045} }

@InProceedings{bergamaschi_liu_single_shot,
  author =	{Bergamaschi, Thiago and Liu, Yunchao},
  title =	{{On Fault Tolerant Single-Shot Logical State Preparation and Robust Long-Range Entanglement}},
  booktitle =	{16th Innovations in Theoretical Computer Science Conference (ITCS 2025)},
  pages =	{16:1--16:9},
  series =	{Leibniz International Proceedings in Informatics (LIPIcs)},
  ISBN =	{978-3-95977-361-4},
  ISSN =	{1868-8969},
  year =	{2025},
  volume =	{325},
  editor =	{Meka, Raghu},
  publisher =	{Schloss Dagstuhl -- Leibniz-Zentrum f{\"u}r Informatik},
  address =	{Dagstuhl, Germany},
  URN =		{urn:nbn:de:0030-drops-226444},
  doi =		{10.4230/LIPIcs.ITCS.2025.16}
}

@article{martinez_simulation,
  title = {Efficient Simulation of Parametrized Quantum Circuits under Nonunital Noise through Pauli Backpropagation},
  author = {Martinez, Victor and Angrisani, Armando and Pankovets, Ekaterina and Fawzi, Omar and Stilck Fran{\c c}a, Daniel},
  journal = {Phys. Rev. Lett.},
  volume = {134},
  issue = {25},
  pages = {250602},
  numpages = {7},
  year = {2025},
  month = {Jun},
  publisher = {American Physical Society},
  doi = {10.1103/j1gg-s6zb},
  url = {https://link.aps.org/doi/10.1103/j1gg-s6zb}
}

@article{oh_classical_simulability_linear_optical,
	Author = {Oh, Changhun},
	Da = {2025/07/29},
	Doi = {10.1038/s41534-025-01041-w},
	Id = {Oh2025},
	Isbn = {2056-6387},
	Journal = {npj Quantum Information},
	Number = {1},
	Pages = {126},
	Title = {Classical simulability of constant-depth linear-optical circuits with noise},
	Ty = {JOUR},
	Url = {https://doi.org/10.1038/s41534-025-01041-w},
	Volume = {11},
	Year = {2025}}

@Article{ibm_evidence_for_utility,
author={Kim, Youngseok
and Eddins, Andrew
and Anand, Sajant
and Wei, Ken Xuan
and van den Berg, Ewout
and Rosenblatt, Sami
and Nayfeh, Hasan
and Wu, Yantao
and Zaletel, Michael
and Temme, Kristan
and Kandala, Abhinav},
title={Evidence for the utility of quantum computing before fault tolerance},
journal={Nature},
year={2023},
month={Jun},
day={01},
volume={618},
number={7965},
pages={500-505},
issn={1476-4687},
doi={10.1038/s41586-023-06096-3},
url={https://doi.org/10.1038/s41586-023-06096-3}
}

@Article{google_threshold,
author= {Acharya, Rajeev and Abanin, Dmitry A.
  and Aghababaie-Beni, Laleh and Aleiner, Igor
  and Andersen, Trond I. and others},
title={Quantum error correction below the surface code threshold},
journal={Nature},
year={2025},
month={Feb},
day={01},
volume={638},
number={8052},
pages={920-926},
issn={1476-4687},
doi={10.1038/s41586-024-08449-y},
url={https://doi.org/10.1038/s41586-024-08449-y}
}

@Article{google_phase_transitions,
author={Morvan, A.
and Villalonga, B.
and Mi, X.
and Mandr{\`a}, S.
and Bengtsson, A.
and Klimov, P. V.
and Chen, Z.
and Hong, S.
and Erickson, C.
and Drozdov, I. K.
and Chau, J.
and Laun, G.
and Movassagh, R.
and Asfaw, A.
and Brand{\~a}o, L. T. A. N.
and Peralta, R.
and Abanin, D.
and Acharya, R.
and Allen, R.
and Andersen, T. I.
and Anderson, K.
and Ansmann, M.
and Arute, F.
and Arya, K.
and Atalaya, J.
and Bardin, J. C.
and Bilmes, A.
and Bortoli, G.
and Bourassa, A.
and Bovaird, J.
and Brill, L.
and Broughton, M.
and Buckley, B. B.
and Buell, D. A.
and Burger, T.
and Burkett, B.
and Bushnell, N.
and Campero, J.
and Chang, H.-S.
and Chiaro, B.
and Chik, D.
and Chou, C.
and Cogan, J.
and Collins, R.
and Conner, P.
and Courtney, W.
and Crook, A. L.
and Curtin, B.
and Debroy, D. M.
and Barba, A. Del Toro
and Demura, S.
and Paolo, A. Di
and Dunsworth, A.
and Faoro, L.
and Farhi, E.
and Fatemi, R.
and Ferreira, V. S.
and Burgos, L. Flores
and Forati, E.
and Fowler, A. G.
and Foxen, B.
and Garcia, G.
and Genois, {\'E}
and Giang, W.
and Gidney, C.
and Gilboa, D.
and Giustina, M.
and Gosula, R.
and Dau, A. Grajales
and Gross, J. A.
and Habegger, S.
and Hamilton, M. C.
and Hansen, M.
and Harrigan, M. P.
and Harrington, S. D.
and Heu, P.
and Hoffmann, M. R.
and Huang, T.
and Huff, A.
and Huggins, W. J.
and Ioffe, L. B.
and Isakov, S. V.
and Iveland, J.
and Jeffrey, E.
and Jiang, Z.
and Jones, C.
and Juhas, P.
and Kafri, D.
and Khattar, T.
and Khezri, M.
and Kieferov{\'a}, M.
and Kim, S.
and Kitaev, A.
and Klots, A. R.
and Korotkov, A. N.
and Kostritsa, F.
and Kreikebaum, J. M.
and Landhuis, D.
and Laptev, P.
and Lau, K.-M.
and Laws, L.
and Lee, J.
and Lee, K. W.
and Lensky, Y. D.
and Lester, B. J.
and Lill, A. T.
and Liu, W.
and Livingston, W. P.
and Locharla, A.
and Malone, F. D.
and Martin, O.
and Martin, S.
and McClean, J. R.
and McEwen, M.
and Miao, K. C.
and Mieszala, A.
and Montazeri, S.
and Mruczkiewicz, W.
and Naaman, O.
and Neeley, M.
and Neill, C.
and Nersisyan, A.
and Newman, M.
and Ng, J. H.
and Nguyen, A.
and Nguyen, M.
and Niu, M. Yuezhen
and O'Brien, T. E.
and Omonije, S.
and Opremcak, A.
and Petukhov, A.
and Potter, R.
and Pryadko, L. P.
and Quintana, C.
and Rhodes, D. M.
and Rocque, C.
and Rosenberg, E.
and Rubin, N. C.
and Saei, N.
and Sank, D.
and Sankaragomathi, K.
and Satzinger, K. J.
and Schurkus, H. F.
and Schuster, C.
and Shearn, M. J.
and Shorter, A.
and Shutty, N.
and Shvarts, V.
and Sivak, V.
and Skruzny, J.
and Smith, W. C.
and Somma, R. D.
and Sterling, G.
and Strain, D.
and Szalay, M.
and Thor, D.
and Torres, A.
and Vidal, G.
and Heidweiller, C. Vollgraff
and White, T.
and Woo, B. W. K.
and Xing, C.
and Yao, Z. J.
and Yeh, P.
and Yoo, J.
and Young, G.
and Zalcman, A.
and Zhang, Y.
and Zhu, N.
and Zobrist, N.
and Rieffel, E. G.
and Biswas, R.
and Babbush, R.
and Bacon, D.
and Hilton, J.
and Lucero, E.
and Neven, H.
and Megrant, A.
and Kelly, J.
and Roushan, P.
and Aleiner, I.
and Smelyanskiy, V.
and Kechedzhi, K.
and Chen, Y.
and Boixo, S.},
title={Phase transitions in random circuit sampling},
journal={Nature},
year={2024},
month={Oct},
day={01},
volume={634},
number={8033},
pages={328-333},
issn={1476-4687},
doi={10.1038/s41586-024-07998-6},
url={https://doi.org/10.1038/s41586-024-07998-6}
}

@article{seddon_magic_monotones,
  title = {Quantifying Quantum Speedups: Improved Classical Simulation From Tighter Magic Monotones},
  author = {Seddon, James R. and Regula, Bartosz and Pashayan, Hakop and Ouyang, Yingkai and Campbell, Earl T.},
  journal = {PRX Quantum},
  volume = {2},
  issue = {1},
  pages = {010345},
  numpages = {42},
  year = {2021},
  month = {Mar},
  publisher = {American Physical Society},
  doi = {10.1103/PRXQuantum.2.010345},
  url = {https://link.aps.org/doi/10.1103/PRXQuantum.2.010345}
}

@article{bouland_fefferman_nirkhe_vazirani_complexity_rcs,
	Author = {Bouland, Adam and Fefferman, Bill and Nirkhe, Chinmay and Vazirani, Umesh},
	Journal = {Nature Physics},
	Number = {2},
	Pages = {159--163},
	Title = {On the complexity and verification of quantum random circuit sampling},
	Volume = {15},
	Year = {2019}}

@article{terhal_divincenzo,
    author = "Terhal, Barbara M. and DiVincenzo, David P.",
    title = "{Adaptive Quantum Computation, Constant Depth Quantum Circuits and Arthur-Merlin Games}",
    eprint = "quant-ph/0205133",
    archivePrefix = "arXiv",
    doi = "10.26421/QIC4.2-5",
    journal = "Quant. Inf. Comput.",
    volume = "4",
    number = "2",
    pages = "134--145",
    year = "2004"
}

@article{de-palma_marvian_rouze_stilck-franca_optimal_transport,
  title = {Limitations of Variational Quantum Algorithms: A Quantum Optimal Transport Approach},
  author = {De Palma, Giacomo and Marvian, Milad and Rouz\'e, Cambyse and Fran{\c c}a, Daniel Stilck},
  journal = {PRX Quantum},
  volume = {4},
  issue = {1},
  pages = {010309},
  numpages = {30},
  year = {2023},
  month = {Jan},
  publisher = {American Physical Society},
  doi = {10.1103/PRXQuantum.4.010309},
  url = {https://link.aps.org/doi/10.1103/PRXQuantum.4.010309}
}

@INPROCEEDINGS{bakshi_high-temperature,
  author={Bakshi, Ainesh and Liu, Allen and Moitra, Ankur and Tang, Ewin},
  booktitle={2024 IEEE 65th Annual Symposium on Foundations of Computer Science (FOCS)}, 
  title={High-Temperature Gibbs States are Unentangled and Efficiently Preparable}, 
  year={2024},
  volume={},
  number={},
  pages={1027-1036},
  doi={10.1109/FOCS61266.2024.00068}}

@inbook{rajakumar_watson_liu,
author = {Joel Rajakumar and James D. Watson and Yi-Kai Liu},
title = {Polynomial-Time Classical Simulation of Noisy IQP Circuits with Constant Depth},
booktitle = {Proceedings of the 2025 Annual ACM-SIAM Symposium on Discrete Algorithms (SODA)},
chapter = {},
pages = {1037-1056},
doi = {10.1137/1.9781611978322.30},
URL = {https://epubs.siam.org/doi/abs/10.1137/1.9781611978322.30},
eprint = {https://epubs.siam.org/doi/pdf/10.1137/1.9781611978322.30},
}

@article{bluvstein_logical,
	Author = {Bluvstein, Dolev and Evered, Simon J. and Geim, Alexandra A. and Li, Sophie H. and Zhou, Hengyun and Manovitz, Tom and Ebadi, Sepehr and Cain, Madelyn and Kalinowski, Marcin and Hangleiter, Dominik and Bonilla Ataides, J. Pablo and Maskara, Nishad and Cong, Iris and Gao, Xun and Sales Rodriguez, Pedro and Karolyshyn, Thomas and Semeghini, Giulia and Gullans, Michael J. and Greiner, Markus and Vuleti{\'c}, Vladan and Lukin, Mikhail D.},
	Journal = {Nature},
	Number = {7997},
	Pages = {58--65},
	Title = {Logical quantum processor based on reconfigurable atom arrays},
	Volume = {626},
	Year = {2024}}

@article{bremner_jozsa_shepherd,
author = {Bremner, Michael J.  and Jozsa, Richard  and Shepherd, Dan J. },
title = {Classical simulation of commuting quantum computations implies collapse of the polynomial hierarchy},
journal = {Proceedings of the Royal Society A: Mathematical, Physical and Engineering Sciences},
volume = {467},
number = {2126},
pages = {459-472},
year = {2011},
doi = {10.1098/rspa.2010.0301},

URL = {https://royalsocietypublishing.org/doi/abs/10.1098/rspa.2010.0301},
eprint = {https://royalsocietypublishing.org/doi/pdf/10.1098/rspa.2010.0301}
,
}

@article{hangleiter_anticoncentration,
  doi = {10.22331/q-2018-05-22-65},
  url = {https://doi.org/10.22331/q-2018-05-22-65},
  title = {Anticoncentration theorems for schemes showing a quantum speedup},
  author = {Hangleiter, Dominik and Bermejo-Vega, Juan and Schwarz, Martin and Eisert, Jens},
  journal = {{Quantum}},
  issn = {2521-327X},
  publisher = {{Verein zur F{\"{o}}rderung des Open Access Publizierens in den Quantenwissenschaften}},
  volume = {2},
  pages = {65},
  month = may,
  year = {2018}
}

@article{bermejo-vega_hangleiter_architectures,
  title = {Architectures for Quantum Simulation Showing a Quantum Speedup},
  author = {Bermejo-Vega, Juan and Hangleiter, Dominik and Schwarz, Martin and Raussendorf, Robert and Eisert, Jens},
  journal = {Phys. Rev. X},
  volume = {8},
  issue = {2},
  pages = {021010},
  numpages = {22},
  year = {2018},
  month = {Apr},
  publisher = {American Physical Society},
  doi = {10.1103/PhysRevX.8.021010},
  url = {https://link.aps.org/doi/10.1103/PhysRevX.8.021010}
}

@INPROCEEDINGS {bergamaschi_chen_liu,
author = { Bergamaschi, Thiago and Chen, Chi-Fang and Liu, Yunchao },
booktitle = { 2024 IEEE 65th Annual Symposium on Foundations of Computer Science (FOCS) },
title = {{ Quantum Computational Advantage with Constant-Temperature Gibbs Sampling }},
year = {2024},
volume = {},
ISSN = {},
pages = {1063-1085},
doi = {10.1109/FOCS61266.2024.00071},
url = {https://doi.ieeecomputersociety.org/10.1109/FOCS61266.2024.00071},
publisher = {IEEE Computer Society},
address = {Los Alamitos, CA, USA},
month =Oct}

@inproceedings{aaronson_arkhipov,
author = {Aaronson, Scott and Arkhipov, Alex},
title = {The computational complexity of linear optics},
year = {2011},
isbn = {9781450306911},
publisher = {Association for Computing Machinery},
address = {New York, NY, USA},
url = {https://doi.org/10.1145/1993636.1993682},
doi = {10.1145/1993636.1993682},
booktitle = {Proceedings of the Forty-Third Annual ACM Symposium on Theory of Computing},
pages = {333–342},
numpages = {10},
location = {San Jose, California, USA},
series = {STOC '11}
}

@misc{ben-or_gottesman_refrigerator,
      title={Quantum Refrigerator}, 
      author={Michael Ben-Or and Daniel Gottesman and Avinatan Hassidim},
      year={2013},
      eprint={1301.1995},
      archivePrefix={arXiv},
      primaryClass={quant-ph},
      url={https://arxiv.org/abs/1301.1995}, 
}

@misc{fujii_noise_threshold,
      title={Noise Threshold of Quantum Supremacy}, 
      author={Keisuke Fujii},
      year={2016},
      eprint={1610.03632},
      archivePrefix={arXiv},
      primaryClass={quant-ph},
      url={https://arxiv.org/abs/1610.03632}, 
}

@INPROCEEDINGS {bouland_fefferman_landau_liu_noise,
author = { Bouland, Adam and Fefferman, Bill and Landau, Zeph and Liu, Yunchao },
booktitle = { 2021 IEEE 62nd Annual Symposium on Foundations of Computer Science (FOCS) },
title = {{ Noise and the Frontier of Quantum Supremacy }},
year = {2022},
volume = {},
ISSN = {},
pages = {1308-1317},
doi = {10.1109/FOCS52979.2021.00127},
url = {https://doi.ieeecomputersociety.org/10.1109/FOCS52979.2021.00127},
publisher = {IEEE Computer Society},
address = {Los Alamitos, CA, USA},
month =Feb}

@phdthesis{yunchao_liu_thesis,
    Author= {Liu, Yunchao},
    Title= {Shallow Quantum Circuits: Algorithms, Complexity, and Fault Tolerance},
    School= {EECS Department, University of California, Berkeley},
    Year= {2024},
    Month= {Aug},
    Url= {http://www2.eecs.berkeley.edu/Pubs/TechRpts/2024/EECS-2024-176.html},
    Number= {UCB/EECS-2024-176},
}

@article{muller-hermes_stilck-franca_wolf_relative_entropy_convergence,
    title={Relative entropy convergence for depolarizing channels},
   volume={57},
   ISSN={1089-7658},
   url={http://dx.doi.org/10.1063/1.4939560},
   DOI={10.1063/1.4939560},
   number={2},
   journal={Journal of Mathematical Physics},
   publisher={AIP Publishing},
   author={Müller-Hermes, Alexander and Stilck França, Daniel and Wolf, Michael M.},
   year={2016},
   month=jan
}

@article{razborov_upper_bounds,
author = {Razborov, Alexander A.},
title = {An upper bound on the threshold quantum decoherence rate},
year = {2004},
issue_date = {May 2004},
publisher = {Rinton Press, Incorporated},
address = {Paramus, NJ},
volume = {4},
number = {3},
issn = {1533-7146},
journal = {Quantum Info. Comput.},
month = may,
pages = {222–228},
numpages = {7}
}

@INPROCEEDINGS{burhman_upper_bounds,
  author={Buhrman, Harry and Cleve, Richard and Laurent, Monique and Linden, Noah and Schrijver, Alexander and Unger, Falk},
  booktitle={2006 47th Annual IEEE Symposium on Foundations of Computer Science (FOCS'06)}, 
  title={New Limits on Fault-Tolerant Quantum Computation}, 
  year={2006},
  volume={},
  number={},
  pages={411-419},
  doi={10.1109/FOCS.2006.50}}

@InProceedings{kempe_upper_bounds,
author="Kempe, Julia
and Regev, Oded
and Unger, Falk
and de Wolf, Ronald",
editor="Aceto, Luca
and Damg{\aa}rd, Ivan
and Goldberg, Leslie Ann
and Halld{\'o}rsson, Magn{\'u}s M.
and Ing{\'o}lfsd{\'o}ttir, Anna
and Walukiewicz, Igor",
title="Upper Bounds on the Noise Threshold for Fault-Tolerant Quantum Computing",
booktitle="Automata, Languages and Programming",
year="2008",
publisher="Springer Berlin Heidelberg",
address="Berlin, Heidelberg",
pages="845--856",
}

@article{plenio_upper_bounds,
doi = {10.1088/1367-2630/12/3/033012},
url = {https://dx.doi.org/10.1088/1367-2630/12/3/033012},
year = {2010},
month = {mar},
publisher = {},
volume = {12},
number = {3},
pages = {033012},
author = {Plenio, M B and Virmani, S},
title = {Upper bounds on fault tolerance thresholds of noisy Clifford-based quantum computers},
journal = {New Journal of Physics}
}

@InProceedings{takahashi_ct-ecs,
  author =	{Takahashi, Yasuhiro and Takeuchi, Yuki and Tani, Seiichiro},
  title =	{{Classically Simulating Quantum Circuits with Local Depolarizing Noise}},
  booktitle =	{45th International Symposium on Mathematical Foundations of Computer Science (MFCS 2020)},
  pages =	{83:1--83:13},
  series =	{Leibniz International Proceedings in Informatics (LIPIcs)},
  ISBN =	{978-3-95977-159-7},
  ISSN =	{1868-8969},
  year =	{2020},
  volume =	{170},
  editor =	{Esparza, Javier and Kr\'{a}l', Daniel},
  publisher =	{Schloss Dagstuhl -- Leibniz-Zentrum f{\"u}r Informatik},
  address =	{Dagstuhl, Germany},
  URL =		{https://drops.dagstuhl.de/entities/document/10.4230/LIPIcs.MFCS.2020.83},
  URN =		{urn:nbn:de:0030-drops-127533},
  doi =		{10.4230/LIPIcs.MFCS.2020.83}
}

@article{stilck-franca_garcia-patron,
	Author = {Stilck Fran{\c c}a, Daniel and Garc{\'\i}a-Patr{\'o}n, Raul},
	Da = {2021/11/01},
	Doi = {10.1038/s41567-021-01356-3},
	Id = {Stilck Fran{\c c}a2021},
	Isbn = {1745-2481},
	Journal = {Nature Physics},
	Number = {11},
	Pages = {1221--1227},
	Title = {Limitations of optimization algorithms on noisy quantum devices},
	Ty = {JOUR},
	Url = {https://doi.org/10.1038/s41567-021-01356-3},
	Volume = {17},
	Year = {2021}}

@article{napp,
  title = {Efficient Classical Simulation of Random Shallow 2D Quantum Circuits},
  author = {Napp, John C. and La Placa, Rolando L. and Dalzell, Alexander M. and Brand\~ao, Fernando G. S. L. and Harrow, Aram W.},
  journal = {Phys. Rev. X},
  volume = {12},
  issue = {2},
  pages = {021021},
  numpages = {32},
  year = {2022},
  month = {Apr},
  publisher = {American Physical Society},
  doi = {10.1103/PhysRevX.12.021021},
  url = {https://link.aps.org/doi/10.1103/PhysRevX.12.021021}
}

@article{cheng_ippoliti,
  title = {Efficient Sampling of Noisy Shallow Circuits Via Monitored Unraveling},
  author = {Cheng, Zihan and Ippoliti, Matteo},
  journal = {PRX Quantum},
  volume = {4},
  issue = {4},
  pages = {040326},
  numpages = {22},
  year = {2023},
  month = {Nov},
  publisher = {American Physical Society},
  doi = {10.1103/PRXQuantum.4.040326},
  url = {https://link.aps.org/doi/10.1103/PRXQuantum.4.040326}
}

@misc{gao_duan,
      title={Efficient classical simulation of noisy quantum computation}, 
      author={Xun Gao and Luming Duan},
      year={2018},
      eprint={1810.03176},
      archivePrefix={arXiv},
      primaryClass={quant-ph},
      url={https://arxiv.org/abs/1810.03176}, 
}

@article{bremner_achieving,
  doi = {10.22331/q-2017-04-25-8},
  url = {https://doi.org/10.22331/q-2017-04-25-8},
  title = {Achieving quantum supremacy with sparse and noisy commuting quantum computations},
  author = {Bremner, Michael J. and Montanaro, Ashley and Shepherd, Dan J.},
  journal = {{Quantum}},
  issn = {2521-327X},
  publisher = {{Verein zur F{\"{o}}rderung des Open Access Publizierens in den Quantenwissenschaften}},
  volume = {1},
  pages = {8},
  month = apr,
  year = {2017}
}

@article{bremner_average,
  title = {Average-Case Complexity Versus Approximate Simulation of Commuting Quantum Computations},
  author = {Bremner, Michael J. and Montanaro, Ashley and Shepherd, Dan J.},
  journal = {Phys. Rev. Lett.},
  volume = {117},
  issue = {8},
  pages = {080501},
  numpages = {5},
  year = {2016},
  month = {Aug},
  publisher = {American Physical Society},
  doi = {10.1103/PhysRevLett.117.080501},
  url = {https://link.aps.org/doi/10.1103/PhysRevLett.117.080501}
}

@article{fujii_tamate,
	Author = {Fujii, Keisuke and Tamate, Shuhei},
	Da = {2016/05/18},
	Doi = {10.1038/srep25598},
	Id = {Fujii2016},
	Isbn = {2045-2322},
	Journal = {Scientific Reports},
	Number = {1},
	Pages = {25598},
	Title = {Computational quantum-classical boundary of noisy commuting quantum circuits},
	Ty = {JOUR},
	Url = {https://doi.org/10.1038/srep25598},
	Volume = {6},
	Year = {2016}}

@INPROCEEDINGS{aharonov_polynomial_1996,
  author={Aharonov, D. and Ben-Or, M.},
  booktitle={Proceedings of 37th Conference on Foundations of Computer Science}, 
  title={Polynomial simulations of decohered quantum computers}, 
  year={1996},
  volume={},
  number={},
  pages={46-55},
  doi={10.1109/SFCS.1996.548463}}

@article{aharonov_quantum_to_classical,
  title = {Quantum to classical phase transition in noisy quantum computers},
  author = {Aharonov, Dorit},
  journal = {Phys. Rev. A},
  volume = {62},
  issue = {6},
  pages = {062311},
  numpages = {13},
  year = {2000},
  month = {Nov},
  publisher = {American Physical Society},
  doi = {10.1103/PhysRevA.62.062311},
  url = {https://link.aps.org/doi/10.1103/PhysRevA.62.062311}
}

@inproceedings{aharonov_polynomial_2023,
author = {Aharonov, Dorit and Gao, Xun and Landau, Zeph and Liu, Yunchao and Vazirani, Umesh},
title = {A Polynomial-Time Classical Algorithm for Noisy Random Circuit Sampling},
year = {2023},
isbn = {9781450399135},
publisher = {Association for Computing Machinery},
address = {New York, NY, USA},
url = {https://doi.org/10.1145/3564246.3585234},
doi = {10.1145/3564246.3585234},
booktitle = {Proceedings of the 55th Annual ACM Symposium on Theory of Computing},
pages = {945–957},
numpages = {13},
location = {Orlando, FL, USA},
series = {STOC 2023}
}

@misc{aharonov_limitations,
      title={Limitations of Noisy Reversible Computation}, 
      author={D. Aharonov and M. Ben-Or and R. Impagliazzo and N. Nisan},
      year={1996},
      eprint={quant-ph/9611028},
      archivePrefix={arXiv},
      primaryClass={quant-ph},
      url={https://arxiv.org/abs/quant-ph/9611028}, 
}

@article{aharonov_ben-or,
author = {Aharonov, Dorit and Ben-Or, Michael},
title = {Fault-Tolerant Quantum Computation with Constant Error Rate},
journal = {SIAM Journal on Computing},
volume = {38},
number = {4},
pages = {1207-1282},
year = {2008},
doi = {10.1137/S0097539799359385},
URL = { https://doi.org/10.1137/S0097539799359385
},
eprint = { https://doi.org/10.1137/S0097539799359385
}
,
}

@article{beveridge_frieze_mcdiarmid,
  title = {Random Minimum Length Spanning Trees in Regular Graphs},
  author = {Beveridge, Andrew and Frieze, Alan M. and McDiarmid, Colin},
  journal = {Combinatorica},
  volume = {18},
  number = {3},
  pages = {311--333},
  year = {1998},
  doi = {10.1007/PL00009825},
  url = {https://doi.org/10.1007/PL00009825}
}

@book{friedli_velenik,
  title = {Statistical Mechanics of Lattice Systems: A Concrete Mathematical Introduction},
  author = {Friedli, Sacha and Velenik, Yvan},
  publisher = {Cambridge University Press},
  year = {2017},
  isbn = {9781107184824},
  doi = {10.1017/9781316882603},
  url = {https://doi.org/10.1017/9781316882603}
}

@article{helmuth_perkins_regts,
  title = {Algorithmic {Pirogov--Sinai} Theory},
  author = {Helmuth, Tyler and Perkins, Will and Regts, Guus},
  journal = {Probability Theory and Related Fields},
  volume = {176},
  number = {3--4},
  pages = {851--895},
  year = {2020},
  doi = {10.1007/s00440-019-00928-y},
  url = {https://doi.org/10.1007/s00440-019-00928-y}
}

@misc{liu_mcginley_schuster_gosset_scrooge,
  title = {Conditional Dependence and {Scrooge} Ensembles in Shallow Random Quantum Circuits},
  author = {Yinchen Liu and Max McGinley and Thomas Schuster and David Gosset},
  year = {2026},
  eprint = {2608.12255},
  archivePrefix = {arXiv},
  primaryClass = {quant-ph},
  url = {https://arxiv.org/abs/2608.12255}
}

@misc{putterman_zlokapa_cotler_thermal_states,
  title = {When Quantum Thermal States Look Classical},
  author = {Harald Putterman and Alexander Zlokapa and Jordan Cotler},
  year = {2026},
  eprint = {2607.28536},
  archivePrefix = {arXiv},
  primaryClass = {quant-ph},
  url = {https://arxiv.org/abs/2607.28536}
}

@article{king_hypercontractivity,
  title = {Hypercontractivity for Semigroups of Unital Qubit Channels},
  author = {King, Christopher},
  journal = {Communications in Mathematical Physics},
  volume = {328},
  number = {1},
  pages = {285--301},
  year = {2014},
  doi = {10.1007/s00220-014-1982-4},
  url = {https://doi.org/10.1007/s00220-014-1982-4}
}

@misc{shao_zhao_cheng_liu_tensor_network,
  title = {Complexity of Tensor Network Simulation for Noisy Quantum Circuits},
  author = {Yuguo Shao and Zishuo Zhao and Song Cheng and Zhengwei Liu},
  year = {2026},
  eprint = {2606.00474},
  archivePrefix = {arXiv},
  primaryClass = {quant-ph},
  url = {https://arxiv.org/abs/2606.00474}
}

@article{trivedi_cirac_simulability_transitions,
  title = {Transitions in Computational Complexity of Continuous-Time Local Open Quantum Dynamics},
  author = {Trivedi, Rahul and Cirac, J. Ignacio},
  journal = {Physical Review Letters},
  volume = {129},
  number = {26},
  pages = {260405},
  year = {2022},
  doi = {10.1103/PhysRevLett.129.260405},
  url = {https://doi.org/10.1103/PhysRevLett.129.260405}
}

@article{wahl_strelchuk_tensor_network_contraction,
  title = {Simulating Quantum Circuits Using Efficient Tensor Network Contraction Algorithms with Subexponential Upper Bound},
  author = {Wahl, Thorsten B. and Strelchuk, Sergii},
  journal = {Physical Review Letters},
  volume = {131},
  number = {18},
  pages = {180601},
  year = {2023},
  doi = {10.1103/PhysRevLett.131.180601},
  url = {https://doi.org/10.1103/PhysRevLett.131.180601}
}

@article{eppstein_miller_teng,
  title = {A Deterministic Linear Time Algorithm for Geometric Separators and Its Applications},
  author = {Eppstein, David and Miller, Gary L. and Teng, Shang-Hua},
  journal = {Fundamenta Informaticae},
  volume = {22},
  number = {4},
  pages = {309--330},
  year = {1995},
  doi = {10.3233/FI-1995-2241},
  url = {https://doi.org/10.3233/FI-1995-2241}
}

@article{kotecky_preiss_1986,
  title = {Cluster Expansion for Abstract Polymer Models},
  author = {Koteck{\'y}, Roman and Preiss, David},
  journal = {Communications in Mathematical Physics},
  volume = {103},
  number = {3},
  pages = {491--498},
  year = {1986},
  doi = {10.1007/BF01211762},
  url = {https://doi.org/10.1007/BF01211762}
}

@article{netocny_redig_2004,
  title = {Large Deviations for Quantum Spin Systems},
  author = {Neto{\v{c}}n{\'y}, Karel and Redig, Frank},
  journal = {Journal of Statistical Physics},
  volume = {117},
  number = {3--4},
  pages = {521--547},
  year = {2004},
  doi = {10.1007/s10955-004-3452-4},
  url = {https://doi.org/10.1007/s10955-004-3452-4}
}

@article{mann_minko_2024,
  title = {Algorithmic Cluster Expansions for Quantum Problems},
  author = {Mann, Ryan L. and Minko, Romy M.},
  journal = {PRX Quantum},
  volume = {5},
  number = {1},
  pages = {010305},
  year = {2024},
  doi = {10.1103/PRXQuantum.5.010305},
  url = {https://doi.org/10.1103/PRXQuantum.5.010305}
}

@misc{chen_rouze_chen_fermions_2025,
  title = {Convergence of the Cumulant Expansion and Polynomial-Time Algorithm for Weakly Interacting Fermions},
  author = {Chen, Hongrui and Rouz{\'e}, Cambyse and Chen, Jielun and Jiang, Jiaqing and Scalet, Samuel O. and Zhan, Yongtao and Chan, Garnet Kin-Lic and Ying, Lexing and Tong, Yu},
  year = {2025},
  eprint = {2512.12010},
  archivePrefix = {arXiv},
  primaryClass = {quant-ph},
  doi = {10.48550/arXiv.2512.12010},
  url = {https://arxiv.org/abs/2512.12010}
}

@misc{zlokapa_quasipolynomial_syk_2026,
  title = {A Rigorous Quasipolynomial-Time Classical Algorithm for {SYK} Thermal Expectations},
  author = {Zlokapa, Alexander},
  year = {2026},
  eprint = {2604.21089},
  archivePrefix = {arXiv},
  primaryClass = {quant-ph},
  doi = {10.48550/arXiv.2604.21089},
  url = {https://arxiv.org/abs/2604.21089}
}

@article{dennis_topological_quantum_memory_2002,
  title = {Topological Quantum Memory},
  author = {Dennis, Eric and Kitaev, Alexei and Landahl, Andrew and Preskill, John},
  journal = {Journal of Mathematical Physics},
  volume = {43},
  number = {9},
  pages = {4452--4505},
  year = {2002},
  doi = {10.1063/1.1499754},
  url = {https://doi.org/10.1063/1.1499754}
}

@article{alicki_thermal_stability_2010,
  title = {On Thermal Stability of Topological Qubit in {Kitaev's} {4D} Model},
  author = {Alicki, R. and Horodecki, M. and Horodecki, P. and Horodecki, R.},
  journal = {Open Systems \& Information Dynamics},
  volume = {17},
  number = {1},
  pages = {1--20},
  year = {2010},
  doi = {10.1142/S1230161210000023},
  url = {https://doi.org/10.1142/S1230161210000023}
}

@article{midha2026beyond,
  title={Beyond belief propagation: Cluster-corrected tensor network contraction with exponential convergence},
  author={Midha, Siddhant and Zhang, Yifan F},
  journal={PRX Quantum},
  volume={7},
  number={3},
  pages={033010},
  year={2026},
  publisher={APS}
}

@article{midha2026algorithmic,
  title={Algorithmic Locality via Provable Convergence in Quantum Tensor Networks},
  author={Midha, Siddhant and Zhang, Yifan F and Malz, Daniel and Abanin, Dmitry A and Gopalakrishnan, Sarang},
  journal={arXiv preprint arXiv:2604.21919},
  year={2026}
}

@article{CE_Mayer1937,
  author = {Mayer, Joseph E.},
  title = {The Statistical Mechanics of Condensing Systems. {I}},
  journal = {The Journal of Chemical Physics},
  volume = {5},
  number = {1},
  pages = {67--73},
  year = {1937},
  doi = {10.1063/1.1749933},
}

@article{CE_shorttimeLRB23,
  title = {Classical Simulation of Short-Time Quantum Dynamics},
  author = {Wild, Dominik S. and Alhambra, \'Alvaro M.},
  journal = {PRX Quantum},
  volume = {4},
  issue = {2},
  pages = {020340},
  numpages = {22},
  year = {2023},
  month = {Jun},
  publisher = {American Physical Society},
  doi = {10.1103/PRXQuantum.4.020340},
  url = {https://link.aps.org/doi/10.1103/PRXQuantum.4.020340}
}

@article{CE_metastab25,
  title = {Theory of Metastable States in Many-Body Quantum Systems},
  author = {Yin, Chao and Surace, Federica M. and Lucas, Andrew},
  journal = {Phys. Rev. X},
  volume = {15},
  issue = {1},
  pages = {011064},
  numpages = {52},
  year = {2025},
  month = {Mar},
  publisher = {American Physical Society},
  doi = {10.1103/PhysRevX.15.011064},
  url = {https://link.aps.org/doi/10.1103/PhysRevX.15.011064},
}

@article{CE_ldpcgap25,
  title = {Low-Density Parity-Check Codes as Stable Phases of Quantum Matter},
  author = {Yin, Chao and Lucas, Andrew},
  journal = {PRX Quantum},
  volume = {6},
  issue = {3},
  pages = {030329},
  numpages = {31},
  year = {2025},
  month = {Aug},
  publisher = {American Physical Society},
  doi = {10.1103/361k-nj4b},
  url = {https://link.aps.org/doi/10.1103/361k-nj4b},
}
